\documentclass[a4paper,UKenglish,cleveref,autoref,thm-restate,numberwithinsect,pdfa]{lipics-v2021}

\usepackage{bussproofs}
\usepackage{multicol}
\usepackage[all]{xy}
\usepackage{fancybox}

\hideLIPIcs

\title{Abstract Framework for All-Path Reachability Analysis toward
  Safety and Liveness Verification (Full Version)} 

\titlerunning{Abstract Framework for All-Path Reachability Analysis
  (Full Version)} 

\author{Misaki Kojima}{Graduate School of Informatics, Nagoya University, Japan
\and \url{https://www.lctrs.jp/kojima/} 
}{kojima@i.nagoya-u.jp}{https://orcid.org/0000-0001-5194-3947}{Grant-in-Aid for JSPS Fellows Grant Number JP24KJ1240}

\author{Naoki Nishida}{Graduate School of Informatics, Nagoya University, Japan
\and \url{https://www.lctrs.jp/nishida/}
}{nishida@i.nagoya-u.ac.jp}{https://orcid.org/0000-0001-8697-4970}{}

\authorrunning{M. Kojima and N. Nishida} 

\Copyright{Misaki Kojima and Naoki Nishida} 



\funding{This research was supported by JSPS KAKENHI Grant Number JP24K02900.}

\acknowledgements{We thank the anonymous reviewers 
of FSCD~2026
for their valuable feedback, which improved the paper.}

\nolinenumbers

\newcommand{\Bfnum}[1]{\textcolor{darkgray}{\rm\sffamily\bfseries #1}}
\newcommand{\bfnum}[1]{(#1)}

\def\cA{\mathcal{A}}

\def\cM{\mathcal{M}}
\def\cO{\mathcal{O}}
\def\cR{\mathcal{R}}
\def\cS{\mathcal{S}}
\def\cT{\mathcal{T}}
\def\cV{\mathcal{V}}
\newcommand{\sort}[2][]{\ifthenelse{\equal{#1}{}}{\mbox{\sl{#2}}}{\mbox{\sl#1{#2}}}}
\newcommand{\symb}[1]{\mathsf{#1}}
\newcommand{\var}[1]{\mathit{#1}}
\newcommand{\Var}{\mathcal{V}ar}
\newcommand{\Constraint}[1]{[ #1 ]}
\newcommand{\CTerm}[2]{\langle #1 \,|\, #2 \rangle}
\newcommand{\NF}[1][]{\ifthenelse{\equal{#1}{}}{\mathit{NF}}{\mathit{NF\!}_{#1}}}

\newcommand{\ReachPred}[2]{#1 \Rightarrow #2}
\newcommand{\Derivative}[1][\cR]{\ifthenelse{\equal{#1}{}}{\partial}{\partial_{#1}}}
\newcommand{\Valid}[3]{#2 \mathrel{\mathrel{\models^\forall_{\mathrm{#1}}}} #3}
\newcommand{\NotValid}[3]{#2 \mathrel{\mathrel{\not\models^\forall_{\mathrm{#1}}}} #3}
\newcommand{\PartiallyValid}[2][\cA]{\Valid{partial}{#1}{#2}}
\newcommand{\NotPartiallyValid}[2][\cA]{\NotValid{partial}{#1}{#2}}
\newcommand{\TotallyValid}[2][\cA]{\Valid{total}{#1}{#2}}
\newcommand{\NotTotallyValid}[2][\cA]{\NotValid{total}{#1}{#2}}
\newcommand{\RuleName}[1]{{\rm \textsc{#1}}}
\newcommand{\InfRule}[3][]{\ifthenelse{\equal{#1}{}}{\displaystyle\frac{~ #2 ~}{~ #3 ~}}{\mbox{\rm\RuleName{#1}} ~\, \displaystyle\frac{~ #2 ~}{~ #3 ~}
}}
\newcommand{\DVP}[1][]{\mathsf{DVP}}
\newcommand{\DVPnew}[1][]{\mathsf{DVP}_+}
\newcommand{\CDVP}[1][]{\mathsf{DVP}}
\newcommand{\CDVPnew}[1][]{\mathsf{DVP}_\bot}
\newcommand{\DCC}[1][]{\mathsf{DCC}}
\newcommand{\DCCnew}[1][]{\mathsf{DCC}_{\ominus}}
\newcommand{\varsl}[1]{\mbox{\sl #1}}
\def\Vbud{V_{\mathit{bud}}}
\def\Vopen{V_{\mathit{open}}}
\def\Vder{V_{\mathrm{Der}}}
\newcommand{\Path}[1]{\to_{\mbox{\scriptsize \RuleName{#1}}}}
\newcommand{\toDer}{\Path{Der}}
\newcommand{\toSubs}{\Path{Subs}}
\newcommand{\toDis}{\Path{Dis}}
\def\cAabcd{\cA_1}
\def\cApet{\cA_2}
\newcommand{\State}[1]{\langle\, #1 \,\rangle}
\newcommand{\StateSet}[1]{\mathit{State}_{#1}}
\newcommand{\ValSem}[2][\cM^{\Sigma}]{ \lfloor\!\!\lfloor\, #2 \,\rfloor\!\!\rfloor_{#1}}
\newcommand{\Eval}[2][\cM^{\Sigma}]{[\![ #2 ]\!]_{#1}}

\newcommand{\CF}{\mathit{CF}}

\def\cSfact{\cS}
\def\SigmaFact{\Sigma}
\def\cMfact{\cM}
\def\cRfact{\cR}

\newenvironment{scprooftree}[1]{\gdef\scalefactor{#1}\begin{center}\proofSkipAmount \leavevmode}{\scalebox{\scalefactor}{\DisplayProof}\proofSkipAmount \end{center} }

\newcommand{\cTheory}{\cT}

\EventEditors{}
\EventNoEds{1}
\EventLongTitle{}
\EventShortTitle{}
\EventAcronym{}
\EventYear{}
\EventDate{}
\EventLocation{}
\EventLogo{}
\SeriesVolume{}
\ArticleNo{}

\begin{document}

\maketitle

\begin{abstract}
	An all-path reachability (APR, for short) predicate over an object set is a pair of a source set and a target set, which are subsets of the object set.
	APR predicates have been defined for abstract reduction systems (ARSs, for short) and then extended to logically constrained term rewrite systems (LCTRSs, for short) as pairs of constrained terms that represent sets of terms modeling configurations, states, etc.
	An APR predicate is partially (or demonically) valid w.r.t.\ a rewrite system if every finite maximal reduction sequence of the system starting from any element in the source set includes an element in the target set.
	Partial validity of APR predicates w.r.t.\ ARSs is defined by means of two inference rules, which can be considered a proof system to construct (possibly infinite) derivation trees for partial validity.
	On the other hand, a proof system for LCTRSs consists of four inference rules, leaving a gap between the inference rules for ARSs and LCTRSs.
	In this paper, we revisit the framework for APR analysis and adapt it to verification of not only safety but also liveness properties.
	To this end, we first reformulate an abstract framework for partial validity w.r.t.\ ARSs so that there is a one-to-one correspondence between the inference rules for partial validity w.r.t.\ ARSs and LCTRSs.
	Secondly, we show how to apply APR analysis to safety verification.
	Thirdly, to apply APR analysis to liveness verification, we introduce a novel stronger validity of APR predicates, called total validity, which requires not only finite but also infinite execution paths to reach target sets.
	Finally, for a partially valid APR predicate with a cyclic-proof tree, we show that the acyclicity of the proof graph obtained from the cyclic-proof tree is a necessary and sufficient condition for total validity.
	The condition implies that if there exists a cyclic-proof tree for an APR predicate, the proof graph of which is acyclic, then the APR predicate is totally valid.
\end{abstract}

\section{Introduction}
\label{sec:intro}

Recently, program verification approaches using \emph{logically constrained term rewrite systems} (LCTRSs, for short)~\cite{KN13frocos} have been extensively investigated~\cite{FKN17tocl,WM18,CL18,NW18vstte,KN18eptcs,KNS19ss,CLB23,MNKS23wst,KNM25jlamp,MFK25}.
LCTRSs are effective models of both functional and imperative programs. For instance, equivalence checking by means of LCTRSs is useful to ensure the correctness of terminating functions (cf.~\cite{FKN17tocl}).
Since the reduction of rewrite systems is in general non-deterministic, rewrite systems are reasonable models of concurrent programs.
A transformation of sequential programs into LCTRSs~\cite{FKN17tocl,KN18eptcs} has been extended to concurrent programs~\cite{KNM25jlamp}.
In addition, a method for runtime-error verification by means of \emph{all-path reachability analysis} of LCTRSs has been developed for, e.g., race and starvation freedom~\cite{KN23jlamp,KN23padl,KN24rp}.

An \emph{all-path reachability} (APR, for short) \emph{predicate} over an object set $A$ is a pair $\ReachPred{P}{Q}$ of a \emph{source} set $P$ and a \emph{target} set $Q$, which are subsets of $A$.
The APR predicate is said to be \emph{partially valid} (or \emph{demonically valid}~\cite[Definition~5]{CL18}) w.r.t.\ a rewrite system $\cR$, the reduction of which is defined over $A$, if every finite \emph{execution path}---a maximal reduction sequence---of $\cR$ starting from any element $s$ in $P$ includes an element in $Q$.
Partial validity of $\ReachPred{P}{Q}$ w.r.t.\ $\cR$ means that
every terminating execution from $P$ eventually reaches $Q$.
The \emph{APR problem} w.r.t.\ $\cR$ is a problem to determine whether given APR predicates are partially valid w.r.t.\ $\cR$ or not.

An \emph{abstract} APR framework, which is a framework for \emph{abstract reduction systems} (ARSs, for short), was first proposed~\cite{CL18}.
In the framework, partial validity of APR predicates w.r.t.\ an ARS $\cA=(A,\to_\cA)$ is defined by the two rules of the inference system $\DVP$ (Demonically Valid Predicate)~\cite{CL18},
which are shown in \Cref{fig:rules-of-DVP}.
Here, the \emph{derivative} $\Derivative[\cA](P)$ is the set of successors of elements in $P$ w.r.t.\ $\to_\cA$ and a set $P$ ($\subseteq A$) is said to be \emph{$\cA$-runnable} if $P\ne \emptyset$ and $P$ does not include any normal form of $\cA$.
To be more precise, the set of partially valid APR predicates w.r.t.\ $\cA$ is defined by the greatest fixed point of the functional of $\DVP$ parameterized by $\cA$.
Rule \RuleName{Subsumption} defines trivially partially valid APR predicates $\ReachPred{P}{Q}$ with $P \subseteq Q$, independent from $\cA$; since $P \subseteq Q$, every execution path starting from $P$ reaches $Q$ because the head element of the path is in $Q$.
Rule \RuleName{Step} first removes the execution paths starting with elements in $P\cap Q$, which are partially valid w.r.t.\ $\cA$, and then generates a subgoal $\ReachPred{\Derivative[\cA](P\setminus Q)}{Q}$ for the tail execution paths of those starting with elements $P\setminus Q$.
Since the inference system $\DVP$ w.r.t.\ $\cA$ can be used to prove partial validity of APR predicates w.r.t.\ $\cA$, it can be considered a proof system to construct (possibly infinite) derivation trees of given APR predicates.

\begin{figure}[t]
	\centering
	$
		\InfRule[Subsumption]{}{\ReachPred{P}{Q}
		}
		~\mbox{if $P \subseteq Q$.}
		\qquad
		\InfRule[Step]{\ReachPred{\Derivative[\cA](P\setminus Q)}{Q}
		}{\ReachPred{P}{Q}}
		~\mbox{if
			$(P \setminus Q)$ is $\cA$-runnable.
		}
	$
	\caption{Inference rules of $\DVP$~\cite{CL18} for partial validity of APR predicates w.r.t.\ an ARS $\cA=(A,\to_{\cA})$, where
		$\Derivative[\cA](P) = \{ t' \in A \mid \exists t \in P.\ t \to_{\cA} t' \}$.
	}
	\label{fig:rules-of-DVP}
\end{figure}

The abstract APR framework has been adapted to LCTRSs, and the proof system $\DCC$ (Demonic Circular Coinduction) for partial validity has been presented~\cite{CL18}.
An APR predicate over a signature $\Sigma$ is a pair $\ReachPred{\CTerm{s}{\phi}}{\CTerm{t}{\psi}}$ of constrained terms $\CTerm{s}{\phi},\CTerm{t}{\psi}$ over $\Sigma$, which represents sets of (ground) terms standing for configurations, states, etc (cf.~\cite{FKN17tocl,KN18eptcs,KN23padl}).
To simplify the discussion, we assume that the two constrained terms have no shared variables.
Note that $\cR$ induces the ARS $(T(\Sigma),\to_\cR)$ for the set $T(\Sigma)$ of ground terms over $\Sigma$, and a constrained term $\CTerm{s}{\phi}$ can be considered the set of ground normalized instances of $s$ by means of substitutions satisfying $\phi$.
The proof system $\DCC$ w.r.t.\ an LCTRS $\cR$ with an underlying theory $\cTheory$ consists of the four inference rules shown in \Cref{fig:rules-of-DCC}, where $\Delta_\cR(\CTerm{s}{\phi})$ is the set of successor constrained terms of $\CTerm{s}{\phi}$ w.r.t.\ constrained narrowing (cf.~\cite[Lemma~3.5]{KN24jip}):
$\Delta_\cR(\CTerm{s}{\phi}) = \{ \CTerm{s[r]_p}{\phi \land (s = s[\ell]_p) \land \varphi} \mid (\ell \to r ~\Constraint{\varphi}) \in \cR, \mbox{$s|_p$ is not a variable} \}$~\cite[Definition~11]{CL18},
where $\ell \to r ~\Constraint{\varphi}$ is renamed so that $\Var(s,\phi) \cap \Var(\ell,r,\varphi) = \emptyset$.
Here, $G$ given in advance is a set of APR predicates to be proved partially valid, e.g., a main goal and other APR predicates used as auxiliary lemmas to prove the main goal.
In applying \RuleName{Circ} to $\ReachPred{\CTerm{s}{\phi}}{\CTerm{t}{\psi}}$,
the APR predicate $\ReachPred{\CTerm{s'}{\phi'}}{\CTerm{t'}{\psi'}}$ in the side condition may be the same as $\ReachPred{\CTerm{s}{\phi}}{\CTerm{t}{\psi}}$.
For the soundness of $\DCC$, \RuleName{Der} has to be applied to all APR predicates in $G$. Note that the application of \RuleName{Subs} is followed by \RuleName{Circ} or either \RuleName{Axiom} or \RuleName{Der}.

\begin{figure}[t]
	\[
		\InfRule[Axiom]{}{\ReachPred{\CTerm{s}{\phi}}{\CTerm{t}{\psi}}
		}
		~\mbox{if $\phi$ is unsatisfiable w.r.t.\ $\cTheory$.}
	\]
	\[
		\InfRule[Subs]{\ReachPred{\CTerm{s}{\phi \land \neg (\exists \vec{x}.\ ((s = t)  \land \psi))}}{\CTerm{t}{\psi}}
		}{\ReachPred{\CTerm{s}{\phi}}{\CTerm{t}{\psi}}
		}
		~\mbox{if $(s = t) \land \phi \land \psi$ is satisfiable w.r.t.\ $\cTheory$,}
	\]
	\qquad
	where $\{\vec{x}\} = \Var(t,\psi) \setminus \Var(s,\phi)$. \[
		\InfRule[Der]{\ReachPred{\CTerm{s_1}{\phi_1}}{\CTerm{t}{\psi}}
			\quad
			\ldots
			\quad
			\ReachPred{\CTerm{s_n}{\phi_n}}{\CTerm{t}{\psi}}
		}{\ReachPred{\CTerm{s}{\phi}}{\CTerm{t}{\psi}}}
		~\mbox{if
			$\CTerm{s}{\phi}$ is $\cR$-runnable,
		}
	\]
	\qquad
	where
	$\Delta_\cR(\CTerm{s}{\phi}) = \{ \CTerm{s_i}{\phi_i} \mid 1 \leq i \leq n \}$
	for some $n>0$.
	\[
		\InfRule[Circ]{\ReachPred{\CTerm{t'}{\phi
					\wedge
					(\exists \vec{x}.\ (s = s') \wedge \phi') \wedge \psi'}}{\CTerm{t}{\psi}}
			\qquad
			\ReachPred{\CTerm{s}{\phi
					\wedge
					\neg
					(\exists \vec{x}.\ (s = s') \wedge \phi') }}{\CTerm{t}{\psi}}
		}{\ReachPred{\CTerm{s}{\phi}}{\CTerm{t}{\psi}}
		}
	\]
	\qquad
	where
	$(\ReachPred{\CTerm{s'}{\phi'}}{\CTerm{t'}{\psi'}}) \in G$
	and
	$\{\vec{x}\} = \Var(s',\phi')$.
	\caption{Inference rules of $\DCC$~\cite{CL18} for partial validity of APR predicates in $G$ of an LCTRS $\cR$ with an underlying theory $\cTheory$.}
	\label{fig:rules-of-DCC}
\end{figure}

The two proof systems $\DVP$ and $\DCC$ have no \emph{one-to-one correspondence} between their inference rules.
The main role of \RuleName{Circ} in $\DCC$ is the introduction of \emph{circularity} to proof trees of $\DCC$ so as to make the trees finite as in \emph{cyclic proofs}~\cite{Bro05}.
Since the initial role of $\DVP$ is to define partial validity, $\DVP$ does not consider circularity and is used to construct possibly infinite proof trees as in $\mbox{LKID}^\omega$~\cite[Chapter~4]{Bro06phd}.
When we do not use \RuleName{Circ} in $\DCC$ and allow the construction of infinite proof trees,
roughly speaking, rules \RuleName{Axiom}, \RuleName{Subs}, and \RuleName{Der} in $\DCC$ correspond to rules \RuleName{Subsumption} and \RuleName{Step} in $\DVP$, while there is no one-to-one correspondence between the inference rules in $\DCC$ and $\DVP$.
The absence of one-to-one correspondence makes the correctness proof of $\DCC$ non-trivial.
If $\DCC$ were clearly an instance of $\DVP$,
then the correctness of $\DCC$ would be immediate from that of $\DVP$;
it suffices to show that rules in $\DCC$ are instances of corresponding rules in $\DVP$.
In addition, the abstract framework is sufficiently useful to investigate APR-based approaches to verification of, e.g., \emph{safety} and \emph{liveness} properties, because concurrent transition systems which can be represented by ARSs are usual models for such properties.

In this paper, we aim to develop an abstract foundation for APR analysis toward runtime-error verification.
To this end, we revisit the framework for partial validity of APR analysis and adapt it to verification of not only safety but also liveness properties such as race and starvation freedom.

To apply APR analysis to runtime-error verification, a weakened but easily implementable variant of $\DCC$ for LCTRSs has been proposed for simpler APR predicates~\cite{KN23jlamp}.
The simplification goes well with the verification of \emph{safety} properties such as race freedom.
On the other hand, it requires the verification of \emph{liveness} properties such as \emph{starvation freedom} to introduce some approximations.
For example, a counter for the waiting time to acquire a semaphore has been introduced to ensure starvation freedom~\cite{KN23padl}:
states in which the counter value exceeds an upper limit specified in advance for verification are approximately considered error states.

The APR framework seems to be well-suited to verification of liveness properties:
We let $P_0$ be either the set of initial states or a set of intermediate states to reach states in a target set $Q$;
then, we attempt to prove partial validity of $\ReachPred{P_0}{Q}$;
the difference operation $P \setminus Q$ leaves states in $P$ that have not yet reached any state in $Q$.
The shortcoming of partial validity is that infinite execution paths do not have to be considered, while it does not matter for verification of safety properties:
Roughly speaking, reduction sequences ending with error states of safety properties are finite execution paths, and infinite execution paths can be excluded.
However, for general runtime-error verification, not only terminating executions but also non-terminating ones must be taken into account, because a runtime error to be verified may happen in a non-terminating execution.
For this problem, in~\cite{KN23jlamp}, simpler APR predicates have been introduced as mentioned above, and a given LCTRS is modified in order to make all finite prefixes of (possibly infinite) execution paths \emph{finite execution paths} of the modified LCTRS.
Unfortunately, the limitation of APR predicates eliminates the good compatibility that exists between the difference operation $P \setminus Q$ and liveness verification.

We first reformulate an abstract framework for partial validity of APR predicates w.r.t.\ ARSs so that there is a one-to-one correspondence between the inference rules for partial validity w.r.t.\ ARSs and LCTRSs.
To be more precise, we reformulate inference rules for partial validity w.r.t.\ ARSs (\Cref{sec:all-path-reachability-for-ARSs}), and adapt a cyclic-proof system for partial validity w.r.t.\ LCTRSs in~\cite[Section~3.2]{KN24jip} to ARSs (\Cref{sec:cyclic-proof-system-for-ARSs}), providing a formal description of proof trees for partial validity of APR predicates.
Secondly, we show how to apply APR analysis to safety verification (\Cref{sec:reduction-of-safety-properties-to-APR}).
Thirdly, to apply APR analysis to liveness verification, we introduce a novel, stronger validity of APR problems, called \emph{total validity}, which requires not only finite but also infinite execution paths to reach target sets (\Cref{subsec:total-validity}).
Note that partial validity is a necessary condition for total validity.
Finally, for a partially valid APR predicate with a cyclic-proof tree, we show that the acyclicity of the \emph{proof graph} obtained from the cyclic-proof tree is a necessary and sufficient condition for total validity, showing how to apply APR analysis to liveness verification (\Cref{subsec:proof-method-for-total-validity}).
All omitted proofs of the claims are provided in the appendix.

As mentioned above, partial validity takes into account finite execution paths only.
To apply the APR framework to runtime-error verification, however, \emph{all} (i.e., not only finite but also infinite) execution paths have to be taken into account.
To this end, we introduce total validity and
show that for a partially valid APR predicate with a proof tree in cyclic-proof style,
there is no cycle in the proof graph obtained from the tree if and only if the APR predicate is totally valid w.r.t.\ $\cA$.
As a consequence, our sufficient condition for total validity of $\ReachPred{P}{Q}$ w.r.t.\ $\cA$ is the existence of an APR proof, the proof graph of which is acyclic.
The cyclic-proof system for APR predicates provides the formal definition of proof graphs, which form the basis for the necessary and sufficient condition of total validity and, in turn, for the reduction of liveness properties to APR predicates.

Since ``demonical validity''~\cite{CL18} (partial validity in this paper) does not consider infinite execution paths and we consider a stronger validity for all execution paths, we renamed ``demonical validity'' to ``partial validity'', which originates from \emph{partial} and \emph{total} correctness of programs.
Notice that total correctness of programs requires termination of programs, but total validity in this paper does not require termination but requires all (possibly infinite) execution paths to reach elements in target sets.

The main contributions of this paper are
(i)
a reformulation of the APR framework for ARSs,
(ii)
a novel notion of validity---total validity---for APR predicates,
and
(iii)
a sufficient condition of total validity for, e.g., APR-based liveness verification.
As demonstrated in this paper for liveness verification, there must be ample scope for further extending the APR framework.
In addition, the APR framework can be ported to other kinds of rewrite systems due to its usefulness in program verification.
However, the abstract framework consisting of the two inference rules in $\DVP$ is too simple to discuss practical applications alongside several rewrite systems, such as LCTRSs.
Note that two formulations already exist for LCTRSs:
the original one~\cite{KN13frocos} and its extension~\cite{CL18}.
Therefore, the abstract APR framework formulated in this paper would be highly significant for further research on APR-based verification.

\section{Preliminaries}
\label{sec:preliminaries}

In this section, we briefly recall all-path reachability of ARSs~\cite{CL18}.
Familiarity with basic notions and notations on ARSs~\cite{BN98,Ohl02} is assumed.

An \emph{abstract reduction system} (ARS, for short) over an object set $A$ is a pair $(A,\to_\cA)$ such that $\to_\cA$ is a binary relation over $A$.
The reflexive closure of $\to_\cA$ is denoted by $\to_\cA^=$,
the transitive closure of $\to_\cA$ is denoted by $\to_\cA^+$, and
the reflexive and transitive closure of $\to_\cA$ is denoted by $\to_\cA^*$.
We denote the set of normal forms (irreducible elements) of $\cA$ by $\NF[\cA]$.
An element $s \in A$ is said to \emph{have a normal form} if there exists a normal form $b \in \NF[\cA]$ such that $a \mathrel{\to_\cA^*} b$.
In the remainder of the paper, we use an ARS $\cA = (A,\to_\cA)$ without notice.

The \emph{derivative} of a set $P$ ($\subseteq A$) is the set $\Derivative[\cA](P) = \{ t \in A \mid \exists s \in P.\ s \to_\cA t \}$.
The set $P$ is called \emph{runnable w.r.t.\ $\cA$} ($\cA$-runnable, for short) if $P\ne \emptyset$ and $P \cap \NF[\cA] = \emptyset$.
By definition, it is clear that $\Derivative[\cA](P)\ne\emptyset$ if and only if $P \setminus \NF[\cA] \ne \emptyset$.
We define execution paths of ARSs.
\begin{definition}[execution path~\cite{CL18}]
	Let $\tau$ be a (possibly infinite) reduction sequence $s_1 \to_\cA s_2 \to_\cA \cdots$  of $\cA$.
	We say that $\tau$ is an \emph{execution path} (of $\cA$) if $\tau$ is maximal, i.e., $\tau$ is either finite ending with an irreducible element or infinite.
\end{definition}
Note that execution paths of $\cA$ are defined coinductively by the inference rules in $\DVP$~\cite{CL18}.

Next, we define APR predicates over $A$.
\begin{definition}[APR predicate~\cite{CL18}]
	An \emph{all-path reachability} (APR, for short) \emph{predicate} over $A$ is a pair $\ReachPred{P}{Q}$ of $P,Q \subseteq A$.\footnote{APR predicates in this paper are called \emph{reachability property} and \emph{reachability formulas} for ARSs and LCTRSs, respectively, in~\cite{CL18}, and \emph{APR problems} in~\cite{KN23jlamp,KN23padl,KN24jip,Kojima24phd}.
		Since we deal with ``all-path reachability'' only, we only use ``all-path reachability predicate'' to unify the terminologies.}
	Note that $Q$ is not restricted to non-empty sets.
\end{definition}
Note that an APR predicate is defined over a set of objects, independent of any ARSs.

Let $\mathit{Rules}$ be a finite set of inference rules over a set $\cO$ of objects, which are of the form $\frac{A_1~\ldots~A_n}{A}$ with $A_1,\ldots,A_n \in \cO$.
$\widehat{\mathit{Rules}}$ stands for the functional of $\mathit{Rules}$ (see e.g.,~\cite[Appendix~A.2]{CL18corr}):
For a set $X$ ($\subseteq \cO$),
$\widehat{\mathit{Rules}}(X) = \{ A \mid \frac{A_1~\ldots~A_n}{A} \in \mathit{Rules},
	A_1,\ldots,A_n \in X \}$.
We denote the greatest fixed point of $\widehat{\mathit{Rules}}$ by $\nu\widehat{\mathit{Rules}}$, i.e., $\nu\widehat{\mathit{Rules}}$ is the greatest set $Y$ ($\subseteq \cO$) such that $Y \supseteq \widehat{\mathit{Rules}}(Y)$.

An APR predicate $\ReachPred{P}{Q}$ over $A$ is said to be \emph{partially valid} (or \emph{demonically valid}~\cite{CL18}) w.r.t.\ $\cA$ if each execution path starting from an element in $P$ either eventually reaches an element in $Q$ or is infinite.
Partial validity is defined coinductively as follows.
\begin{definition}[partial validity~\cite{CL18}]
	\label{def:DVP}
	An APR predicate $\ReachPred{P}{Q}$ over $A$ is said to be \emph{partially valid} w.r.t.\ $\cA$, written as $\PartiallyValid[\cA]{\ReachPred{P}{Q}}$,
	if $(\ReachPred{P}{Q}) \in \nu\widehat{\DVP[\cA]}$, where $\DVP$ consists of
	the rules \RuleName{Subsumption} and \RuleName{Step} in \Cref{fig:rules-of-DVP}.
	Note that $\DVP$ implicitly takes $\cA$ as a parameter.
	We write $\NotPartiallyValid[\cA]{\ReachPred{P}{Q}}$ if $\ReachPred{P}{Q}$ is not partially valid w.r.t.\ $\cA$.
\end{definition}

By definition, APR predicates trivially have the following properties for partial validity.
\begin{proposition}
	\label{prop:basic-properties-of-partially-valid-APR-problems}
	Let $P,P',Q,Q',R \subseteq A$.
	Then, all of the following hold:
	\begin{enumerate}
		\renewcommand{\labelenumi}{(\arabic{enumi})}
		\leftskip=1ex
		\item if
		      $\PartiallyValid{\ReachPred{P}{Q}}$ and $\PartiallyValid{\ReachPred{P'}{Q'}}$,
		      then
		      $\PartiallyValid{\ReachPred{(P\cup P')}{(Q\cup Q')}}$,
		\item if
		      $\PartiallyValid{\ReachPred{(P\cup P')}{Q}}$,
		      then
		      $\PartiallyValid{\ReachPred{P}{Q}}$ and $\PartiallyValid{\ReachPred{P'}{Q}}$,
		\item if
		      $\PartiallyValid{\ReachPred{P}{Q}}$,
		      then
		      $\PartiallyValid{\ReachPred{P}{Q'}}$
		      for any $Q'$ such that $Q \subseteq Q' \subseteq A$,
		\item if
		      $\PartiallyValid{\ReachPred{P}{Q}}$ and $\PartiallyValid{\ReachPred{Q}{R}}$,
		      then
		      $\PartiallyValid{\ReachPred{P}{R}}$,
		      and
		\item $\PartiallyValid{\ReachPred{P}{\emptyset}}$
		      if and only if
		      no element in $P$ has a normal form (i.e., $\{ t \in \NF[\cA] \mid \exists s \in P.\ s \mathrel{\to_\cA^*} t \} = \emptyset$).

	\end{enumerate}
\end{proposition}

For an APR predicate $\ReachPred{P}{Q}$, regarding partial validity,
\Cref{prop:basic-properties-of-partially-valid-APR-problems}~\bfnum{1}--\bfnum{2} implies
soundness and completeness of splitting $P$ into $n$ sets $P_1,\ldots,P_n$ (i.e., $P = P_1 \cup \cdots \cup P_n$):
$\PartiallyValid{\ReachPred{(P_1\cup\cdots\cup P_n)}{Q}}$
if and only if
$\PartiallyValid{\ReachPred{P_i}{Q}}$ for all $1 \leq i \leq n$.
On the other hand, splitting the target set $Q$ into $Q_1,\ldots,Q_n$ (i.e., $Q=Q_1\cup\cdots\cup Q_n$) is sound by \Cref{prop:basic-properties-of-partially-valid-APR-problems}~\bfnum{1}, but not complete in general:
if
$\PartiallyValid{\ReachPred{P}{Q_i}}$ for all $1 \leq i \leq n$,
then,
$\PartiallyValid{\ReachPred{P}{(Q_1\cup\cdots\cup Q_n)}}$,
but
the other direction does not hold in general (see \Cref{ex:cAabcd} below).

The inference rules in $\DVP$ can be considered a proof system for partial validity of APR predicates, i.e.,
$\PartiallyValid{\ReachPred{P}{Q}}$
if and only if there exists a (possibly infinite) proof tree obtained from $\ReachPred{P}{Q}$ by applying the rules of $\DVP$.

\begin{example}
	\label{ex:cAabcd}
	Consider the ARS $\cAabcd=(\{\symb{a},\symb{b},\symb{c},\symb{d}\},\to_{\cAabcd})$ such that
	${\to_{\cAabcd}} = \{
		(\symb{a},\symb{b}),
		(\symb{a},\symb{d}), \linebreak
		(\symb{b},\symb{a}),
		(\symb{b},\symb{c})
		\}$.
	We have that $\NotPartiallyValid[\cAabcd]{\ReachPred{\{\symb{a}\}}{\{\symb{c}\}}}$
	because
	the finite execution path $\symb{a} \to_{\cAabcd} \symb{d}$ does not reach $\symb{c}$, and thus
	$(\ReachPred{\{\symb{a}\}}{\{\symb{c}\}}) \notin \nu\widehat{\DVP[\cAabcd]}$ w.r.t.\ $\cAabcd$.
	In fact, only the rule \RuleName{Step} is applicable, and we obtain the following stuck incomplete proof tree:
	\begin{prooftree}
		\AxiomC{$\ReachPred{\{\symb{b},\symb{d}\}}{\{\symb{c}\}}$}
		\RightLabel{\small(\RuleName{Step})}
		\UnaryInfC{$\ReachPred{\{\symb{a}\}}{\{\symb{c}\}}$}
	\end{prooftree}
	On the other hand, we have that $\PartiallyValid[\cAabcd]{\ReachPred{\{\symb{a}\}}{\{\symb{c},\symb{d}\}}}$
	because $(\ReachPred{\{\symb{a}\}}{\{\symb{c},\symb{d}\}}) \in \nu\widehat{\DVP[\cAabcd]}$ w.r.t.\ $\cAabcd$, i.e., we have the following infinite proof tree:
	\begin{prooftree}
		\AxiomC{$\vdots$}
		\RightLabel{\small (\RuleName{Step})}
		\UnaryInfC{$\ReachPred{\{\symb{b},\symb{d}\}}{\{\symb{c},\symb{d}\}}$}
		\RightLabel{\small(\RuleName{Step})}
		\UnaryInfC{$\ReachPred{\{\symb{a},\symb{c}\}}{\{\symb{c},\symb{d}\}}$}
		\RightLabel{\small(\RuleName{Step})}
		\UnaryInfC{$\ReachPred{\{\symb{b},\symb{d}\}}{\{\symb{c},\symb{d}\}}$}
		\RightLabel{\small(\RuleName{Step})}
		\UnaryInfC{$\ReachPred{\{\symb{a}\}}{\{\symb{c},\symb{d}\}}$}
	\end{prooftree}
	To prove that $\PartiallyValid[\cAabcd]{\ReachPred{\{\symb{b},\symb{d}\}}{\{\symb{c},\symb{d}\}}}$, by \Cref{prop:basic-properties-of-partially-valid-APR-problems}~\bfnum{2}, we can split $\ReachPred{\{\symb{b},\symb{d}\}}{\{\symb{c},\symb{d}\}}$ into
	$\ReachPred{\{\symb{b}\}}{\{\symb{c},\symb{d}\}}$
	and
	$\ReachPred{\{\symb{d}\}}{\{\symb{c},\symb{d}\}}$, both of which are partially valid w.r.t.\ $\cAabcd$.
	On the other hand,
	$\ReachPred{\{\symb{b},\symb{d}\}}{\{\symb{c},\symb{d}\}}$ cannot be split into
	$\ReachPred{\{\symb{b},\symb{d}\}}{\{\symb{c}\}}$
	and
	$\ReachPred{\{\symb{b},\symb{d}\}}{\{\symb{d}\}}$, the former of which is not partially valid w.r.t.\ $\cAabcd$.
\end{example}

\section{Reformulation of Inference Rules for ARSs}
\label{sec:all-path-reachability-for-ARSs}

In this section, we reformulate the proof system $\DVP$ for partial validity w.r.t.\ ARSs to establish a one-to-one correspondence between the inference rules for ARSs and LCTRSs.

By abstracting the inference rules in $\DCC$ for LCTRSs in \Cref{fig:rules-of-DCC}, we obtain inference rules for ARSs, which have one-to-one correspondence with the inference rules for LCTRSs.

\begin{definition}[$\DVPnew$]
	\label{def:DVPnew}
	We define $\DVPnew$ consisting of the inference rules in \Cref{fig:rules-of-DVPnew}.
\end{definition}
In applying \RuleName{Subs} and \RuleName{Der} of $\DVPnew$ to $\ReachPred{P}{Q}$, we can freely split the results of $P \setminus Q$ and $\Derivative[\cA](P)$ into $n$ sets $P_1,\ldots,P_n$, respectively.
If $P_i=\emptyset$ for \RuleName{Subs} or \RuleName{Der}, then the generated subgoal $\ReachPred{P_i}{Q}$ can immediately be proved by \RuleName{Axiom}.
To avoid such redundant splits, we assume that
\begin{itemize}
	\item ``if $n > 1$ then
	      $\emptyset \subset P_i \subseteq P$ for all $1 \leq i \leq n$'' for \RuleName{Subs},
	      and
	\item ``$P_i \ne \emptyset$ for all $1 \leq i \leq n$'' for \RuleName{Der}.
\end{itemize}
Note that for an $\cA$-runnable set $P$, we have that $\Derivative[\cA](P)\ne\emptyset$ for \RuleName{Der}, because
$P \setminus \NF[\cA]  \ne \emptyset$.
Note also that $P_i$ and $P_j$ ($i \ne j$) may overlap, i.e., $P_i \cap P_j \ne \emptyset$.

\begin{figure}[t]
	\begin{minipage}[t]{.4\textwidth}
		\[
			\InfRule[Axiom]{}{\ReachPred{P}{Q}
			}
			~\mbox{if $P = \emptyset$.}
		\]
	\end{minipage}
	\begin{minipage}[t]{.59\textwidth}
		\[
			\InfRule[Subs]{\ReachPred{P_1}{Q}
				\quad
				\ldots
				\quad
				\ReachPred{P_n}{Q}
			}{\ReachPred{P}{Q}
			}
			~\mbox{if $P \cap Q  \ne \emptyset$,}
		\]
		\qquad
		where $P \setminus Q = P_1 \cup \cdots \cup P_n$ for some $n > 0$.
	\end{minipage}

	\begin{minipage}{.99\textwidth}
		\[
			\InfRule[Der]{\ReachPred{P_1}{Q}
				\quad
				\ldots
				\quad
				\ReachPred{P_n}{Q}
			}{\ReachPred{P}{Q}
			}
			~\mbox{if $P \cap Q = \emptyset$ and $P$ is $\cA$-runnable,}
		\]
		\qquad
		where $\Derivative[\cA](P) = P_1 \cup \cdots \cup P_n$ for some $n > 0$.
	\end{minipage}

	\caption{Rules of $\DVPnew$ for proving partial validity of APR predicates of an ARS $\cA$.}
	\label{fig:rules-of-DVPnew}
\end{figure}

By definition, the side conditions of rules in $\DVPnew$ are orthogonal, and thus,
for each APR predicate, at most one rule in $\DVPnew$ is applicable (cf.~\cite{KN24jip}).
\begin{restatable}{proposition}{PropUniquenessOfApplicableRulesInDVPnew}
	\label{prop:uniqueness-of-applicable-rules-in-DVPnew}
	Let $P,Q \subseteq A$.
	Then, at most one rule in $\DVPnew$ is applicable to the APR predicate $\ReachPred{P}{Q}$.
\end{restatable}
Note that the way of applying \RuleName{Subs} and \RuleName{Der} in $\DVPnew$ is not unique, because multiple decompositions of resulting subgoals are possible.

For any ARS, $\DVP$ and $\DVPnew$ coincide.
\begin{restatable}{theorem}{ThmEquivalenceOfDVPandDVPnew}
	\label{thm:equivalence-of-DVP-and-DVPnew}
	For any ARS $\cA$, $\nu\widehat{\DVP[\cA]} = \nu\widehat{\DVPnew[\cA]}$.
\end{restatable}
By \Cref{thm:equivalence-of-DVP-and-DVPnew}, $\DVP$ and $\DVPnew$ define the same partially valid APR predicates and they are equivalent as proof systems for partial validity of APR predicates.
On the other hand, $\DVPnew$ has more flexibility than $\DVP$ w.r.t.\ the split of source sets of generated subgoals, in addition to the one-to-one correspondence with $\DCC$ in terms of description.

The disproof criterion for LCTRSs~\cite{KN24jip} is formulated for ARSs as follows.
\begin{proposition}
	\label{prop:disproof-criterion}
	Let $\ReachPred{P}{Q}$ be an APR predicate over $A$.
	If $P \cap Q = \emptyset$
	and
	$P \cap \NF[\cA] \ne \emptyset$,
	then
	$\NotPartiallyValid{\ReachPred{P}{Q}}$.
\end{proposition}
\begin{proof}
	Assume that $P \cap Q = \emptyset$ and $P \cap \NF[\cA] \ne \emptyset$.
	Then, there exists a normal form $s \in P \cap \NF[\cA]$ such that $s \notin Q$, and thus we have a finite execution path $s$ that does not reach $Q$.
	Therefore, we have that $\NotPartiallyValid{\ReachPred{P}{Q}}$.
\end{proof}

By definition, it is clear that $P\ne \emptyset$ and $P \cap \NF[\cA] \ne \emptyset$ if and only if $P \ne \emptyset$ and $P$ is not $\cA$-runnable.
\Cref{prop:disproof-criterion} implies an inference rule for disproof (cf.~\cite{KN24jip}).
\begin{definition}[$\CDVPnew$]
	\label{def:CDVP}
	We define $\CDVPnew$ consisting of the inference rules in $\DVPnew$ and the following rule for disproof:
	\[
		\InfRule[Dis]{\bot
		}{\ReachPred{P}{Q}}
		~\mbox{if
			$P \cap Q = \emptyset$,
			$P \ne \emptyset$,
			and
			$P$ is not $\cA$-runnable.
		}
	\]
	Note that $\bot$ is used as a special case---the \emph{invalid} one---of APR predicates.
\end{definition}
By definition, the side conditions of rules in $\CDVPnew$ are orthogonal and exhaustive, and thus,
for each APR predicate, exactly one rule in $\CDVPnew$ is applicable (cf.~\cite{KN24jip}).
\begin{restatable}{proposition}{PropUniquenessOfApplicableRulesInCDVPnew}
	\label{prop:uniqueness-of-applicable-rules-in-CDVPnew}
	Let $P,Q \subseteq A$.
	Then, exactly one rule in $\CDVPnew$ is applicable to the APR predicate $\ReachPred{P}{Q}$.
\end{restatable}

\section{Cyclic Proof System for APR Problems of ARSs}
\label{sec:cyclic-proof-system-for-ARSs}

In this section, we adapt the cyclic-proof system for partial validity of APR predicates w.r.t.\ LCTRSs~\cite{KN24jip} to the abstract APR framework reformulated in \Cref{sec:all-path-reachability-for-ARSs}.
To this end, we first revisit \RuleName{Circ} in \Cref{fig:rules-of-DCC}, and then define APR pre-proofs for partial validity w.r.t.\ ARSs.

In proving partial validity of APR predicates using the inference rules such as $\CDVPnew$, we would like to construct finite proof trees if possible.
In the APR framework, a \emph{circularity} rule, \RuleName{Circ}, has been introduced~\cite[Definition~12]{CL18} (see \Cref{fig:rules-of-DCC}).
Roughly speaking, the rule can be considered as a composition of ``circularity'' in \emph{cyclic proofs}~\cite{Bro05} and ``cut'' in \emph{sequent calculus}.
Rule \RuleName{Circ} in \Cref{fig:rules-of-DCC} for LCTRSs is formulated for ARSs as follows:
\[
	\InfRule[Circ]{\ReachPred{Q'}{Q}
		\qquad
		\ReachPred{(P \setminus P')}{Q}
	}{\ReachPred{P}{Q}
	}
	~\mbox{if $(\ReachPred{P'}{Q'}) \in G$,}
\]
where $G$ is a given set of APR predicates over $A$.
Note that $\ReachPred{P'}{Q'}$ may be the same as $\ReachPred{P}{Q}$, while \RuleName{Der} has to be applied to $\ReachPred{P'}{Q'}$ somewhere for soundness.
Note also that $Q'$ of the subgoal $\ReachPred{Q'}{Q}$ may be improved by replacing it with $Q''$ such that $\{ t \in Q' \mid \exists s \in P\cap P'.\ s \mathrel{\to_\cA^*} t \} \subseteq Q'' \subseteq Q'$.
As stated in \Cref{sec:intro}, $G$ is assumed to be given in advance as a set of APR predicates to be proved partially valid.
In practice, we start with a single APR predicate to be proved partially valid, $G$ is the empty set;
when we apply \RuleName{Der} to an APR predicate, we add the predicates to $G$;
for soundness, for each APR predicate to which \RuleName{Circ} is applied, there must exist an APR predicate in $G$ to which \RuleName{Der} has already been applied, i.e.,
\RuleName{Der} must be applied to all APR predicates in $G$.

From the viewpoint of practical use, we do not consider ``cut'' because, as for the use of ``cut'' in proof theories, it is not so easy to split an APR predicate $\ReachPred{P}{Q}$ into appropriate ones $\ReachPred{P}{R}$ and $\ReachPred{R}{Q}$ for some set $R$;
deriving $R$ is similar to a human deriving a lemma.
Then, the circularity rule without ``cut'' is formulated for ARSs as follows:
\[
	\InfRule[Cyc]{}{\ReachPred{P}{Q}
	}
	~\mbox{if $P \subseteq P'$ and $Q \supseteq Q'$ for some $(\ReachPred{P'}{Q'}) \in G$,}
\]
where $G$ is a given set of APR predicates over $A$.
To distinguish the circularity rule without ``cut'' from \RuleName{Circ}, we named the former \RuleName{Cyc}, which originates from ``cyclic proofs''.
The role of \RuleName{Cyc} is not only ``circularity'' but also ``generalization'' of APR predicates:
$\ReachPred{P'}{Q'}$ is more general than $\ReachPred{P}{Q}$
in the sense that
if
$\PartiallyValid{\ReachPred{P'}{Q'}}$,
then
$\PartiallyValid{\ReachPred{P}{Q}}$
(\Cref{prop:basic-properties-of-partially-valid-APR-problems}~\bfnum{1}--\bfnum{2}).
Viewed in this light,
if $P=P'$ and $Q=Q'$, then \RuleName{Cyc} is just ``circularity'' in cyclic proofs,
and
otherwise, \RuleName{Cyc} generalizes $\ReachPred{P}{Q}$ to $\ReachPred{P'}{Q'}$, which is included in $G$ to be proved partially valid as, e.g., a more general subgoal for the main goal.
Note that ``cut'' is formulated as follows:
\[
	\InfRule[Cut]{\ReachPred{P}{Q'}
		\qquad
		\ReachPred{Q'}{Q}
	}{\ReachPred{P}{Q}
	}
\]
As in other proof systems, \RuleName{Cut} must be powerful but needs a heuristic for automation.
As a first step, we leave the introduction of \RuleName{Cut} to the proof system below as future work.

In~\cite{KN24jip}, a simpler proof system for partial validity w.r.t.\ LCTRSs has been formulated in the \emph{cyclic-proof} style~\cite{Bro05}.
In the proof system, the circularity rule is not explicitly used, but the \emph{bud}-\emph{companion} relationship of cyclic proofs is used instead in proof trees.
We formulate the cyclic-proof system for partial validity w.r.t.\ ARSs.
\begin{definition}[derivation tree of {$\CDVP$}]
	\label{def:APR-derivation-tree}
	An \emph{APR derivation tree} w.r.t.\ $\cA=(A,\to_\cA)$ is a finite tree $\cT=(V,\varsl{a},\varsl{r},\varsl{c})$ such that
	\begin{itemize}
		\item $V$ is a finite set of nodes,
		\item $\varsl{a}$ is a total mapping from $V$ to the set of APR predicates over $A$, which includes $\bot$ as an APR predicate, \item $\varsl{r}$ is a partial mapping from $V$ to $\CDVPnew$,
		\item $\varsl{c}$ is a partial mapping from $V$ to $V^*$ (we write $\varsl{c}_j(v)$ for the $j$-th component of $\varsl{c}(v)$) which is the $j$-th child of $v$,
		      and
		\item for all nodes $v \in V$, $\varsl{c}_j(v)$ is defined just in case $\varsl{r}(v)$ is a rule with $k$ premises ($1 \leq j \leq k)$, and $\frac{\varsl{a}(\varsl{c}_1(v)) ~ \ldots ~ \varsl{a}(\varsl{c}_k(v))}{\varsl{a}(v)}$ is an instance of rule $\varsl{r}(v)$,
		      where if $k=0$, then $\varsl{c}(v)=\epsilon$.
	\end{itemize}
	Note that a node $v \in V$ is a leaf if and only if either $\varsl{c}(v)=\epsilon$ or $\varsl{c}(v)$ is undefined.
	For a node $v$ with $\varsl{a}(v) = (\ReachPred{P}{Q})$, $\varsl{a}_L(v)$ and $\varsl{a}_R(v)$ denote $P$ and $Q$, respectively.
	A leaf $v$ of $\cT$ is said to be \emph{closed} if $\varsl{c}(v)=\epsilon$ or $\varsl{a}(v)=\bot$.
	A leaf $v$ of $\cT$ is said to be \emph{open} if it is not closed, i.e., $\varsl{c}(v)\ne\epsilon$ and $\varsl{a}(v)\ne\bot$.
	The set of open leaves of $\cT$ is denoted by $\Vopen$.
\end{definition}
Note that any node $v$ with $\varsl{a}(v)\ne \bot$ must have a rule attached, i.e., $\varsl{r}(v)$ must be defined.
By introducing the circularity relationship into APR derivation trees, we define APR pre-proofs.
\begin{definition}[APR pre-proof]
	\label{def:APR-pre-proof}
	An \emph{APR pre-proof} for an APR predicate $\ReachPred{P}{Q}$ w.r.t.\ $\cA$ is a pair $(\cT,\xi)$ of an APR derivation tree $\cT=(V,\varsl{a},\varsl{r},\varsl{c})$ (with $v_0$ the root node) and a partial mapping $\xi$ from $\Vopen$ to $V\setminus \Vopen$
	such that
	$\varsl{a}(v_0) = (\ReachPred{P}{Q})$,
	and
	for any open leaf $v$, if $\xi(v)$ is defined, then
	$\xi(v)$ is a node of $\cT$ such that $\varsl{a}(v)=\varsl{a}(\xi(v))$\footnote{\label{fnt:relaxation-of-bud-companion-relationships}
		From the viewpoint of \RuleName{Cyc}, the condition ``$\varsl{a}(v)=\varsl{a}(\xi(v))$'' can be relaxed to
		``$\varsl{a}_L(v) \subseteq \varsl{a}_L(\xi(v))$''.
		However, as a starting point, we do not use the relaxed condition and leave it as future work.}
	and $\varsl{r}(\xi(v)) = \mbox{\rm\RuleName{Der}}$.
	An open leaf $v$ with $\xi(v)$ defined is called a \emph{bud} node of $\cT$, and the node $\xi(v)$ is called a \emph{companion} of $v$.
	We denote the set of bud nodes in $V$ by $\Vbud$ ($\subseteq \Vopen$).
	The APR pre-proof $(\cT,\xi)$ is said to be \emph{closed} if
	any leaf $v \in V$ is either closed or a bud node.
	The APR pre-proof $(\cT,\xi)$ is said to be \emph{open} if $\cT$ is not closed, i.e., there exists an open leaf that is not a bud.
\end{definition}
Note that a companion does not have to be an ancestor of its bud nodes.
\Cref{tbl:APR-pre-proof-nodes} illustrates how $\varsl{a}$, $\varsl{r}$, $\varsl{c}$, and $\xi$ are defined for a node $v \in V$ regarding an APR pre-proof $((V,\varsl{a},\varsl{r},\varsl{c}),\xi)$.
Unlike the usual definition of companions in cyclic-proofs, for a bud node $v$, we required the additional condition ``$\varsl{r}(\xi(v)) = \mbox{\rm\RuleName{Der}}$''.
The reason for requiring the additional condition will be explained later.

\begin{table}
	\caption{How mappings $\varsl{a}$, $\varsl{r}$, $\varsl{c}$, and $\xi$ are defined for a node $v \in V$ regarding an APR pre-proof $((V,\varsl{a},\varsl{r},\varsl{c}),\xi)$.}
	\label{tbl:APR-pre-proof-nodes}
	\centering
	\begin{tabular}{|c|c|c|c|c|c|}
		\hline
		$\varsl{a}(v)$     & $\varsl{r}(v)$   & $\varsl{c}(v)$                               & open/closed     & $\xi$                                             \\
		\hline\hline
		                   & \RuleName{Axiom} & $\epsilon$                                   & closed leaf     &                                                   \\
		\cline{2-4}
		                   & \RuleName{Subs}  &                                              &                 &                                                   \\
		\cline{2-2}
		$\ReachPred{P}{Q}$ & \RuleName{Der}   & \raisebox{7pt}[0pt]{$v_1\ldots v_n$ ($n>0$)} & (internal node) & \raisebox{7pt}[0pt]{(not in the domain of $\xi$)} \\
		\cline{2-3}
		                   & \RuleName{Dis}   & $v_1$ such that $\varsl{a}(v_1)=\bot$        &                 &                                                   \\
		\cline{2-5}
		                   & undefined        & undefined                                    & open leaf       & defined or undefined                              \\
		\hline
		$\bot$             & undefined        & undefined                                    & closed leaf     & (not in the domain of $\xi$)                      \\
		\hline
	\end{tabular}
\end{table}

\begin{example}
	\label{ex:APR-pre-proof-of-cAabcd}
	Let us consider the ARS $\cAabcd$ in \Cref{ex:cAabcd} again.
	Let $v_0,v_1,\ldots,v_5$ be nodes such that
	\begin{itemize}
		\item $\varsl{a}(v_0) = (\ReachPred{\{\symb{a}\}}{\{\symb{c},\symb{d}\}})$, $\varsl{r}(v_0)=\mbox{\rm\RuleName{Der}}$, $\varsl{c}(v_0) = v_1$,
		\item $\varsl{a}(v_1) = (\ReachPred{\{\symb{b},\symb{d}\}}{\{\symb{c},\symb{d}\}})$, $\varsl{r}(v_1)=\mbox{\rm\RuleName{Subs}}$, $\varsl{c}(v_1) = v_2$,
		\item $\varsl{a}(v_2) = (\ReachPred{\{\symb{b}\}}{\{\symb{c},\symb{d}\}})$, $\varsl{r}(v_2)=\mbox{\rm\RuleName{Der}}$, $\varsl{c}(v_2) = v_3v_4$,
		\item $\varsl{a}(v_3) = (\ReachPred{\{\symb{a}\}}{\{\symb{c},\symb{d}\}})$, neither $\varsl{r}(v_3)$ nor $\varsl{c}(v_3)$ is defined,
		\item $\varsl{a}(v_4) = (\ReachPred{\{\symb{c}\}}{\{\symb{c},\symb{d}\}})$, $\varsl{r}(v_4)=\mbox{\rm\RuleName{Subs}}$, $\varsl{c}(v_4) = v_5$,
		\item $\varsl{a}(v_5) = (\ReachPred{\emptyset}{\{\symb{c},\symb{d}\}})$, $\varsl{r}(v_5)=\mbox{\rm\RuleName{Axiom}}$, $\varsl{c}(v_5) = \epsilon$,
		      and
		\item $\xi(v_3) = v_0$.
	\end{itemize}
	Then, $((\{v_0,v_1,\ldots,v_5\},\varsl{a},\varsl{r},\varsl{c}),\xi)$ is an APR pre-proof for $\ReachPred{\{\symb{a}\}}{\{\symb{c},\symb{d}\}}$ w.r.t.\ $\cAabcd$,
	which is visualized in \Cref{fig:cAabcd-proof}~(a).
\end{example}

\begin{figure}[t]
	\begin{tabular}{@{}c@{\quad}c@{}}
		\begin{minipage}{.55\textwidth}
			\begin{prooftree}\AxiomC{$v_3: ~ \ReachPred{\{\symb{a}\}}{\{\symb{c},\symb{d}\}}$\dag}
				\AxiomC{}
				\RightLabel{\small(\RuleName{Axiom})}
				\UnaryInfC{$v_5: ~ \ReachPred{\emptyset}{\{\symb{c},\symb{d}\}}$}
				\RightLabel{\small(\RuleName{Subs})}
				\UnaryInfC{$v_4: ~ \ReachPred{\{\symb{c}\}}{\{\symb{c},\symb{d}\}}$}
				\RightLabel{\small(\RuleName{Der})}
				\BinaryInfC{$v_2: ~ \ReachPred{\{\symb{b}\}}{\{\symb{c},\symb{d}\}}$}
				\RightLabel{\small(\RuleName{Subs})}
				\UnaryInfC{$v_1: ~ \ReachPred{\{\symb{b},\symb{d}\}}{\{\symb{c},\symb{d}\}}$}
				\RightLabel{\small(\RuleName{Der})}
				\UnaryInfC{$v_0: ~ \ReachPred{\{\symb{a}\}}{\{\symb{c},\symb{d}\}}$\dag}
			\end{prooftree}
		\end{minipage}
		 &
		\begin{minipage}{.12\textwidth}
			\[
				\xymatrix@R=8pt@C=5pt{
				&& *+[o][F]{v_5} \\
				&& *+[o][F]{v_4} \ar[u] \\
				& *+[o][F]{v_2}  \ar[ur] \ar@(ul,ul)[dd]\\
				& *+[o][F]{v_1} \ar[u] \\
				& *+[o][F]{v_0} \ar[u] \\
				}
			\]
		\end{minipage}
		\\[48pt]
		(a) APR pre-proof
		 & (b) proof graph \\
	\end{tabular}
	\caption{An APR pre-proof and its proof graph for the APR predicate $\ReachPred{\{\symb{a}\}}{\{\symb{c},\symb{d}\}}$, where $\dag$ indicates the bud-companion relationship.}
	\label{fig:cAabcd-proof}
\end{figure}

The proof graph obtained from a closed APR pre-proof is defined as follows.
\begin{definition}[proof graph]
	\label{def:proof-graph}
	The \emph{proof graph} of a closed APR pre-proof $(\cT,\xi)$ with $\cT=(V,\varsl{a},\varsl{r},\varsl{c})$ is a directed graph $(V',E)$ obtained from $\cT$ by identifying each bud node and its companion:
	\begin{itemize}
		\item $V' = V \setminus \Vopen$, where
		      each node in $V$ is associated with the APR predicate $\varsl{a}(v)$,
		      and
		\item
		      $(v,v') \in E$ if and only if
		      there exists a node $v'' \in V$ such that
		      $v''$ appears in $\varsl{c}(v)$,
		      if $v''$ is a bud node then $v'$ is a companion of $v''$, and otherwise, $v''=v'$.
	\end{itemize}
	For an edge $(v,v') \in E$ and a rule name $X \in \{ \mbox{\rm\RuleName{Axiom}}, \mbox{\rm\RuleName{Subs}}, \mbox{\rm\RuleName{Der}}, \mbox{\rm\RuleName{Dis}} \}$, we write $v \Path{\it X} v'$
	if $\varsl{r}(v)=\mbox{\it X}$.
\end{definition}

\begin{example}
	The proof graph of $((\{v_0,v_1,\ldots,v_5\},\varsl{a},\varsl{r},\varsl{c}),\xi)$ in \Cref{ex:APR-pre-proof-of-cAabcd} (i.e., \Cref{fig:cAabcd-proof}~(a)) is illustrated in \Cref{fig:cAabcd-proof}~(b).
\end{example}

By definition, nodes of proof graphs trivially have the following properties.
\begin{proposition}
	\label{prop:basic-properties-of-proof-graphs}
	Let $(\cT,\xi)$ be a closed APR pre-proof,
	$\cT= ((V,\varsl{a},\varsl{r},\varsl{c})$,
	$(V,E)$ be the proof graph of $(\cT,\xi)$,
	$v,v' \in V$,
	and
	$X \in  \{ \mbox{\rm\RuleName{Axiom}}, \mbox{\rm\RuleName{Subs}}, \mbox{\rm\RuleName{Der}}, \mbox{\rm\RuleName{Dis}} \}$.
	Then, all of the following statements hold:
	\begin{enumerate}
		\renewcommand{\labelenumi}{(\arabic{enumi})}
		\leftskip=1ex
		\item if $(v,v') \in E$, then $\varsl{a}_R(v) = \varsl{a}_R(v')$,
		\item if $v \toSubs v'$, then $(\varsl{a}_L(v) \setminus \varsl{a}_R(v)) \supseteq \varsl{a}_L(v')$,
		\item if $v \toDer v'$, then both of the following hold:
		      \begin{itemize}
			      \item for any $s \in \varsl{a}_L(v)$, there exist a node $v'' \in V$ and an element $t \in \varsl{a}_L(v'')$ such that
			            $v \toDer v''$ and $s \to_\cA t$,
			            and
			      \item for any $t \in \varsl{a}_L(v')$, there exists an element $s \in \varsl{a}_L(v)$ such that $s \to_\cA t$,
		      \end{itemize}
		\item if $v \toDis v'$, then $\varsl{a}(v') = \bot$,
		      and
		\item if $\varsl{r}(v) = \mbox{\rm\RuleName{Axiom}}$, then
		      $\varsl{a}_L(v) = \emptyset$ and $v$ has no outedge.
	\end{enumerate}
\end{proposition}

Cyclic proofs satisfy the \emph{global trace condition}~\cite{Bro05}.
The application of rules of \emph{inductive predicates} is the measure of coinduction in the cyclic-proof setting, which ensures the soundness of cyclic proofs.
Roughly speaking, the global trace condition requires every infinite \emph{trace}---a sequence of atomic formulas in premises sets of sequents---to follow an infinite path of the proof graph, which goes through infinitely many edges corresponding to the application of rules of inductive predicates.
In our case, such edges correspond to those from nodes with $\mbox{\rm\RuleName{Der}}$ and for soundness of an APR proof, the APR proof has to satisfy the global trace condition, i.e., every infinite path of the proof graph has to go through node with \RuleName{Der} infinitely often.
To make proof graphs implicitly satisfy the global trace condition, we require companions to be nodes with $\mbox{\rm\RuleName{Der}}$;
the form of APR predicates can be considered the same as that of sequents, and each APR predicate has exactly one set corresponding to an atomic formula of the premises of sequents;
viewed in this light, traces correspond to paths of proof graphs of closed APR pre-proofs;
every cycle goes through companions infinitely often, and thus every infinite trace follows an infinite path that goes through infinitely many nodes with $\mbox{\rm\RuleName{Der}}$.

The requirement to the bud-companion relationship---companions have \RuleName{Der} attached---does not lose generality.
Let $((V,\varsl{a},\varsl{r},\varsl{c}),\xi)$ be an APR pre-proof for an APR predicate $\ReachPred{P}{Q}$ w.r.t.\ $\cA$, $v$ be a bud node, and $v'$ be a companion of $v$ (i.e., $\xi(v)=v'$).
If $\varsl{r}(v') \ne \mbox{\rm\RuleName{Der}}$, then we can drop the bud-companion relationship between $v$ and $v'$, obtaining another APR pre-proof for $\ReachPred{P}{Q}$ as follows:
\begin{itemize}
	\item if $\varsl{r}(v') \in \{\mbox{\rm\RuleName{Axiom}}, \mbox{\rm\RuleName{Dis}}\}$, then we let
	      $\varsl{r}(v) = \varsl{r}(v')$ and $\varsl{c}(v)=\varsl{c}(v')$,
	      and
	\item if $\varsl{r}(v')=\mbox{\rm\RuleName{Subs}}$ and $\varsl{c}(v')=v'_1\ldots v'_n$, then
	      we introduce new $n$ nodes $v_1,\ldots,v_n$ to $V$ and let
	      $\varsl{r}(v) = \mbox{\rm\RuleName{Subs}}$,
	      $\varsl{c}(v)=v_1\ldots v_n$,
	      and
	      $\varsl{a}(v_i) = \varsl{a}(v'_i)$ and $\xi(v_i)=v'_i$ for all $i \in \{1,\ldots,n\}$.
\end{itemize}
While the second case above introduces a new bud-companion relationship to the transformed APR pre-proof, the repetition of the above transformation halts because $v'_i$ is a child of $v'$ such that $\varsl{r}(v)=\mbox{\rm\RuleName{Subs}}$ and thus $\varsl{r}(v'_i) \in  \{\mbox{\rm\RuleName{Axiom}},\mbox{\rm\RuleName{Der}},\mbox{\rm\RuleName{Dis}}\}$.

We now define APR proofs and disproofs as APR pre-proofs satisfying certain conditions.
\begin{definition}[APR proof and disproof]
	\label{def:APR-proof}
	An APR pre-proof $((V,a,r,p),\xi)$ for an APR predicate $\ReachPred{P}{Q}$ w.r.t.\ $\cA$ is called
	an \emph{APR proof} if the domain of $\xi$ is $\Vopen$ (i.e., all open nodes are bud nodes) and there is no node $v \in V$ such that $\varsl{r}(v)= \mbox{\rm\RuleName{Dis}}$.
	The APR pre-proof $((V,a,r,p),\xi)$ is called
	an \emph{APR disproof} if there exists a node $v \in V$ such that $\varsl{r}(v)= \mbox{\rm\RuleName{Dis}}$.
\end{definition}
Note that a closed APR pre-proof is either an APR proof or an APR disproof.
Note also that an APR disproof may have open nodes that are not bud nodes.

\begin{example}
	The APR pre-proof in \Cref{fig:cAabcd-proof}~(a) is an APR proof for $\ReachPred{\{\symb{a}\}}{\{\symb{c},\symb{d}\}}$.
	The APR pre-proof in \Cref{fig:cAabcd-disproof} is an APR disproof for $\ReachPred{\{\symb{a}\}}{\{\symb{c}\}}$.
\end{example}

\begin{figure}[t]
	\begin{minipage}{.95\textwidth}
		\begin{prooftree}\AxiomC{$\bot$}
			\RightLabel{\small(\RuleName{Dis})}
			\UnaryInfC{$\ReachPred{\{\symb{b},\symb{d}\}}{\{\symb{c}\}}$}
			\RightLabel{\small(\RuleName{Der})}
			\UnaryInfC{$\ReachPred{\{\symb{a}\}}{\{\symb{c}\}}$}
		\end{prooftree}
	\end{minipage}
	\caption{An APR pre-proof for $\ReachPred{\{\symb{a}\}}{\{\symb{c}\}}$, which is an APR disproof.}
	\label{fig:cAabcd-disproof}
\end{figure}

APR proofs and disproofs are sound.
\begin{restatable}{theorem}{ThmSoundnessOfAPRProofs}
	\label{thm:soundness-of-APR-proofs}
	Let $\ReachPred{P}{Q}$ be an APR predicate over $A$.
	Then, both of the following statements hold:
	\begin{enumerate}
		\renewcommand{\labelenumi}{(\arabic{enumi})}
		\leftskip=1ex
		\item if there exists an APR proof for $\ReachPred{P}{Q}$ w.r.t.\ $\cA$, then
		      $\PartiallyValid{\ReachPred{P}{Q}}$,
		      and
		\item if there exists an APR disproof of $\ReachPred{P}{Q}$ w.r.t.\ $\cA$, then
		      $\NotPartiallyValid{\ReachPred{P}{Q}}$.
	\end{enumerate}
\end{restatable}

When $\{ t \in A \mid \exists s \in P.\ s \to_\cA^* t \}$ is finite,
the APR predicate $\ReachPred{P}{Q}$ has either an APR proof or an APR disproof.
In addition, if the side conditions of rules in $\CDVPnew$ are decidable,
then partial validity of $\ReachPred{P}{Q}$ w.r.t.\ $\cA$ is decidable.

\begin{restatable}{theorem}{ThmFinitenessAndDecidabilityOfAPRPredicate}
	\label{thm:finiteness-and-decidability-of-APR-predicates}
	Let $\ReachPred{P}{Q}$ be an APR predicate over $A$,
	and $\Derivative[\cA]^*(P) = \{ t \in A \mid \exists s \in P.\ s \to_\cA^* t \}$.
	Then, both of the following statements hold:
	\begin{enumerate}
		\renewcommand{\labelenumi}{(\arabic{enumi})}
		\leftskip=1ex
		\item
		      If $\Derivative[\cA]^*(P)$ is finite, then there exists either an APR proof or disproof for $\ReachPred{P}{Q}$,
		      and
		\item if $\Derivative[\cA]^*(P)$ is finite and the emptiness, intersection emptiness, and $\cA$-runnability problems for $A$ are decidable, then partial validity of $\ReachPred{P}{Q}$ w.r.t.\ $\cA$
		      is decidable.
	\end{enumerate}
\end{restatable}

When we simultaneously prove two or more APR predicates, i.e., a set $G$ of APR predicates, to be partially valid, the notion of APR derivation trees and pre-proofs can be extended to forests:
We consider a set of APR pre-proofs for APR predicates in $G$;
the bud-companion relationship is allowed between nodes in different derivation trees.

The proof system $\DCC$ for LCTRSs~\cite{CL18} and its weakened variant~\cite{KN23jlamp} are instances of $\DVPnew$ and $\CDVPnew$, respectively,
and APR pre-proofs w.r.t.\ ARSs, which provide a concrete method for constructing APR (dis)proofs for APR predicates w.r.t.\ LCTRSs, are used as a foundation for both systems.
For the page limitation, in the rest of this section, we give an informal proof for $\DCC$ being an instance of $\DVPnew$.
A formal proof can be seen in the appendix.

Let $\cR$ be an LCTRS over a signature $\Sigma$ with an underlying theory $\cTheory$~\cite[Section~3]{CL18}.
LCTRS $\cR$ induces the ARS $(T(\Sigma),\to_\cR)$.
A \emph{constrained term} $\CTerm{s}{\phi}$ consisting of a term $s$ and a constraint $\phi$ represents the set of ground normalized instances $s\gamma$ w.r.t.\ $\phi$, where $\gamma$ is a ground normalized substitution and $\phi\gamma$ is evaluated to $\symb{true}$.
In light of this, we deal with constrained terms as the set of ground instances, and pairs of constrained terms can be considered APR problems over $T(\Sigma)$.
We show that each rule in $\DCC$ is an instance of the corresponding rule in $\DVPnew$ w.r.t.\ the ARS $(T(\Sigma),\to_\cR)$.
Let us consider an APR problem $\ReachPred{\CTerm{s}{\phi}}{\CTerm{t}{\psi}}$.
\begin{itemize}
	\item
	      By definition, it is clear that $\phi$ is unsatisfiable if and only if $\CTerm{s}{\phi}$ is the emptyset.
	      Thus, \RuleName{Axiom} in $\DCC$ is an instance of \RuleName{Axiom} in $\DVPnew$.
	\item
	      For a constrained rewrite rule $\ell \to r ~\Constraint{\varphi}$ in $\cR$,
	      satisfiability of $(s = t) \land \phi \land \psi$ means that
	      $\CTerm{s}{\phi} \cap \CTerm{t}{\psi} \ne \emptyset$, and
	      the constrained term $\CTerm{s}{\phi \land \neg (\exists \vec{x}.\ ((s = t)  \land \psi))}$
	      represents
	      the set $\CTerm{s}{\phi} \setminus \CTerm{t}{\psi}$, where $\{\vec{x}\} = \Var(t,\psi) \setminus(s,\phi)$.
	      Thus, \RuleName{Subs} in $\DCC$ is an instance of \RuleName{Subs} in $\DVPnew$.
	\item
	      We have that $\Delta_\cR(\CTerm{s}{\phi}) = \Derivative[(T(\Sigma),\to_\cR)](\CTerm{s}{\phi})$~\cite[Theorem~1]{CL18}.
	      By definition, $\cR$-runnability of $\CTerm{s}{\phi}$ is defined by $(T(\Sigma),\to_\cR)$-runnability of the set $\CTerm{s}{\phi}$.
	      Thus, \RuleName{Der} in $\DCC$ is an instance of \RuleName{Der} in $\DVPnew$.
	\item
	      To remove the ``cut'' function from \RuleName{Circ} in $\DCC$,
	      both ``$\CTerm{t'}{\psi'} \subseteq \CTerm{t}{\psi}$''
	      and
	      ``$\CTerm{s}{\phi} \subseteq \CTerm{s'}{\phi'}$''
	      are necessary.
	      Under these conditions, the premises are partially valid and can thus be dropped.
	      Then, a simplified variant without ``cut'' is obtained as follows:
	      \label{eq:Cyc-for-LCTRS}
	      \begin{equation*}\InfRule[Cyc]{}{\ReachPred{\CTerm{s}{\phi}\!}{\!\CTerm{t}{\psi}}
		      }
		      ~\mbox{if
			      $\CTerm{s}{\phi} \,{\subseteq}\, \CTerm{s'\!}{\phi'}$
			      and
			      $\CTerm{t}{\psi} \,{\supseteq}\, \CTerm{t'\!}{\psi'}$
			      for some $(\ReachPred{\CTerm{s'\!}{\phi'}\!}{\!\CTerm{t'\!}{\psi'}}) \,{\in}\, G$.}
	      \end{equation*}
	      By definition, the above \RuleName{Cyc} is an instance of \RuleName{Cyc} in $\DVPnew$.
\end{itemize}

\section{Reduction of Safety Properties to APR predicates}
\label{sec:reduction-of-safety-properties-to-APR}

Informally speaking, \emph{safety properties} specify that ``something bad never happens''~\cite[Section~3.3.2]{BK08}.
As a safety property w.r.t.\ a system and a set $E$ of error states, we consider the property that
any (possibly non-terminating) execution of the system never reaches any error state in $E$.
This kind of safety property is formulated as a problem for ARSs as follows:
\begin{description}
	\item[Instance] An ARS $\cA=(A,\to_\cA)$ and sets $P,E$ ($\subseteq A$)
	\item[Question] Is there no (possibly infinite) execution path $s_0 \to_\cA s_1 \to_\cA \cdots$ such that $s_0 \in P$ and
	      $s_i \in E$ for some $i \geq 0$?
\end{description}
Note that not all error states are, in general, irreducible.
This problem is the same as \emph{non-reachability} from $P$ to $E$ w.r.t.\ $\cA$.
On the other hand, this problem can be reduced to the APR problem~\cite{KN23jlamp}.
In this section, we revisit the reduction
of the non-reachability problem
to partial validity of APR predicates w.r.t.\ LCTRSs, and then reformulate it for ARSs.

Note that in general, the negation of partial validity of the APR predicate $\ReachPred{P}{E}$ is not equivalent to non-reachability from $P$ to $E$:
\begin{itemize}
	\item The negation of partial validity of $\ReachPred{P}{E}$ is that
	      there exists a finite execution path starting from a state in $P$, which does not include any error state in $E$, and
	\item non-reachability from $P$ to $E$ is that
	      there exists \emph{no} execution path starting from a state in $P$, which does not include any error state in $E$.
\end{itemize}

We first revisit the reduction to APR predicates w.r.t.\ LCTRSs in~\cite{KN24jip}.
For an LCTRS $\cR$ w.r.t.\ error states represented by constrained terms $\CTerm{e_1}{\psi_1},\ldots,\CTerm{e_k}{\psi_k}$,
the reduction proceeds as follows:
\begin{enumerate}
	\item
	      Fresh constants $\symb{success}$ and $\symb{error}$ are introduced to the signature of $\cR$.
	\item
	      Constrained rewrite rules to reduce any state, including an error state, to $\symb{success}$ are added to $\cR$.
	\item
	      Constrained rewrite rules $e_i \to \symb{error} ~ \Constraint{\psi_i}$ with $1 \leq i \leq k$ are added to $\cR$.
	\item
	      The non-reachability problem w.r.t.\ the error states $E$ is reduced to
	      the APR predicate $\ReachPred{I}{\{\symb{success}\}}$, where $I$ is the set of ground terms for initial states of a system. \end{enumerate}
Note that the introduced constants are normal forms of the modified LCTRS.
The role of $\symb{success}$ is to make all prefixes of any reduction sequence from an initial state a finite execution path ending with $\symb{success}$.
If an initial execution path reaches an error state represented by some $\CTerm{e_i}{\psi_i}$, then there exists a finite execution path ending with $\symb{error}$.

We now adapt the above approach to ARSs.
Usually, error states are irreducible, and thus, we consider only irreducible error states.
If an error state is reducible, then, as in the approach above for LCTRSs, we introduce a fresh constant such as $\symb{error}$, together with the additional reduction from the error states to $\symb{error}$, which is considered a dummy error state in the modified system.

Let us reconsider to reduce non-reachability (from $P$ to $E$ w.r.t.\ an ARS $\cA$) to a partially valid APR predicate $\ReachPred{P}{Q}$.\footnote{We do not consider to reduce to APR predicates of the form $\ReachPred{P'}{Q}$ such that $P' \ne P$, because we should consider all execution paths starting from $P$, and thus the APR predicate $\ReachPred{P}{P''}$ for some $P'' \subseteq P'$ needs to be partially valid w.r.t.\ $\cA$.
	When we choose an APR predicate of the form $\ReachPred{P}{Q}$, we do not have to find appropriate sets $P',P''$.}
For the reduced partially valid APR predicate, we do not have to take into account infinite execution paths, because any execution path ending with an error state is a finite execution path.
If $Q \cap E \ne \emptyset$, then the reduction is not sound, and thus $Q$ should have no error state in $E$: $Q \cap E = \emptyset$.
For the partial validity, all normal forms that are reachable from $P$ and are not error states should be included in $Q$:
$\{ t \in \NF[\cA] \setminus E \mid \exists s \in P.\ s \to_\cA^* t \} \subseteq Q$.
If $Q$ includes a reducible state, then the reduction may not be sound:
An error state may be reachable from the reducible state.
Thus, $Q$ should be a set of normal forms w.r.t.\ $\cA$.
We call such $\ReachPred{P}{Q}$ a \emph{safety APR predicate for $E$ w.r.t.\ $\cA$}.

\begin{definition}[safety APR predicate for error states]
	Let $E\subseteq \NF[\cA]$ be a set of error states.
	An APR predicate $\ReachPred{P}{Q}$ is called a \emph{safety predicate for $E$ w.r.t.\ $\cA$} if
	$Q \cap E = \emptyset$,
	$\{ t \in \NF[\cA] \setminus E \mid \exists s \in P.\ s \to_\cA^* t \} \subseteq Q$,
	and
	$Q \subseteq \NF[\cA]$.
\end{definition}

Safety APR predicates can be used for verification of safety properties.

\begin{restatable}{theorem}{ThmPartiallyValidSafetyProperties}
	\label{thm:partially-valid-safety-properties}
	Let $P, Q\subseteq A$, $E \subseteq \NF[\cA]$, and $\ReachPred{P}{Q}$ be a safety APR predicate for $E$ w.r.t.\ $\cA$. Then,
	$\PartiallyValid{\ReachPred{P}{Q}}$
	if and only if
	there is no finite execution path of $\cA$ that starts with an element in $P$ and includes an element in $E$.
\end{restatable}

\begin{figure}[t]
	\centering
	\scalebox{.95}{$\xymatrix@C=5ex{
		\mathit{PG}_0:~~  \ar[r] & \ovalbox{$\symb{noncrit}_0$} \ar[rrrr]^{\mbox{$\langle b_0,x \rangle:=\langle \symb{true},1 \rangle$}} &&&& \ovalbox{$\symb{wait}_0$} \ar[rrrr]^{\mbox{$x=0 \,\lor\,\lnot b_{1}$}} &&&& \ovalbox{$\symb{crit}_0$} \ar@/^12pt/[llllllll]^{\mbox{$b_0:=\symb{false}$}}
		}
	$}
	\\[5pt]
	\scalebox{.95}{$\xymatrix@C=5ex{
		\mathit{PG}_1:~~  \ar[r] & \ovalbox{$\symb{noncrit}_1$} \ar[rrrr]^{\mbox{$\langle b_1,x \rangle:=\langle \symb{true},0 \rangle$}} &&&& \ovalbox{$\symb{wait}_1$} \ar[rrrr]^{\mbox{$x=1 \,\lor\,\lnot b_{0}$}} &&&& \ovalbox{$\symb{crit}_1$} \ar@/^12pt/[llllllll]^{\mbox{$b_1:=\symb{false}$}}
		}
	$}
	\caption{Program graph $\mathit{PG}_i$ ($i=0,1$) for Peterson's mutual exclusion algorithm.}
	\label{fig:PG-Peterson}
\end{figure}

\begin{example}\label{example:peterson-ec}
	Let us consider the \emph{program graphs} $\mathit{PG}_0, \mathit{PG}_1$ in \Cref{fig:PG-Peterson} for \emph{Peterson's mutual exclusion algorithm}~\cite[Example~2.25]{BK08}.
	Two processes $P_0,P_1$ specified by $\mathit{PG}_0, \mathit{PG}_1$, respectively, share Boolean variables $b_0,b_1$ and an integer variable $x$;
	the Boolean variables $b_i$ indicates that $P_i$ wants to enter the critical section $\symb{crit}_i$;
	the integer variable $x$ stores the identifier of the process that has priority for the critical section at that time.
	A state of the \emph{asynchronous integer transition system} (AITS, for short) with shared variables~\cite{BK08} consisting of $\mathit{PG}_0$ and $\mathit{PG}_1$ is
	a tuple of a location of $\mathit{PG}_0$, a location of $\mathit{PG}_1$, and an \emph{assignment} from variables $b_0,b_1,x$ to values.
	The initial states of the AITS are
	$\State{\symb{noncrit}_0,\symb{noncrit}_1,\symb{false},\symb{false},m}$ with $m\in\{0,1\}$.
	Let $\StateSet{i}$ be the set $\{\symb{noncrit}_i,\symb{wait}_i,\symb{crit}_i\}$ ($i\in\{0,1\}$), $\mathbb{B}$ be the set $\{\symb{true},\symb{false}\}$, and $\cS$ be the set of assignments for variables $b_0,b_1:\sort{bool}$ and $x:\sort{int}$, i.e.,
	$\cS = \{ \{b_0\mapsto v_0, b_1 \mapsto v_1, x \mapsto m \} \mid v_0,v_1 \in \mathbb{B}, ~ m \in \mathbb{Z} \}$.
	We denote a state $\State{p_0,p_1,\sigma}$ by $\State{p_0,p_1,\sigma(b_0),\sigma(b_1),\sigma(x)}$, where $\sigma \in \cS$.
	The AITS consisting of $P_0$ and $P_1$ is represented by the following ARS:
	\[
		\cApet =
		(\{ \State{p_0,p_1,\sigma} \mid p_0 \in \StateSet{0}, \, p_1 \in \StateSet{1}, \, \sigma \in \cS \},
		\to_{\cApet})
	\]
	such that
	\begin{itemize}
		\item $\State{\symb{noncrit}_0,p_1,v_0,v_1,m} \to_{\cApet} \State{\symb{wait}_0,p_1,\symb{true},v_1,1}$,
		\item $\State{\symb{wait}_0,p_1,v_0,v_1,m} \to_{\cApet} \State{\symb{crit}_0,p_1,\symb{true},v_1,m}$
		      if $m=0$ or $v_1 = \symb{false}$,
		\item $\State{p_0,\symb{crit}_1,v_0,v_1,m} \to_{\cApet} \State{p_0,\symb{noncrit}_1,v_,\symb{false},m}$,
		\item $\State{p_0,\symb{noncrit}_1,v_0,v_1,m} \to_{\cApet} \State{p_0,\symb{wait}_1,v_0,\symb{true},0}$,
		\item $\State{p_0,\symb{wait}_1,v_0,v_1,m} \to_{\cApet} \State{p_0,\symb{crit}_1,v_0,\symb{true},m}$
		      if $m=1$ or $v_0 = \symb{false}$,
		      and
		\item $\State{p_0,\symb{crit}_1,v_0,v_1,m} \to_{\cApet} \State{p_0,\symb{noncrit}_1,v_0,\symb{false},m}$,
	\end{itemize}
	where
	$p_i \in \StateSet{i}$ with $i\in\{0,1\}$,
	$v_0,v_1 \in \mathbb{B}$,
	and
	$m \in \mathbb{Z}$.

	Let us consider the \emph{race freedom} of mutual exclusion---the two processes do not enter their critical sections simultaneously---for the AITS $\cApet$.
	There is no normal form reachable from an initial state.
	This property will be examined later, together with race freedom.
	The error states are $\State{\symb{crit}_0,\symb{crit}_1,\sigma}$ with $\sigma \in \cS$.
	Since all the error states are reducible w.r.t.\ $\cApet$,
	we introduce a fresh element $\symb{error}$ as a dummy error state to be verified, and add
	the reduction from the original error states to $\symb{error}$:
	$\State{\symb{crit}_0,\symb{crit}_1,\sigma} \to_{\cApet} \symb{error}$, where $\sigma \in \cS$.
	We let $\cApet'$ be the modified ARS $(\{ \State{p_0,p_1,\sigma} \mid p_0 \in \StateSet{0}, \, p_1 \in \StateSet{1}, \, \sigma \in \cS \},\to_{\cApet'})$, where
	${\to_{\cApet'}} = {\to_{\cApet}} \cup \{ (\State{\symb{crit}_0,\symb{crit}_1,\sigma}, \symb{error}) \mid \sigma \in \cS \}$.
	The race freedom is reduced to the APR predicate
	$
		(1) ~
		\ReachPred{\{ \State{\symb{noncrit}_0,\symb{noncrit}_1,\symb{false},\symb{false},m} \mid m \in \{0,1\} \}}{\emptyset}
	$.
	We obtain an APR proof illustrated in \Cref{fig:APR-proof-rf}, and thus, by \Cref{thm:soundness-of-APR-proofs},
	the APR predicate~(1) is partially valid w.r.t.\ $\cApet'$.
	Therefore,
	by \Cref{thm:partially-valid-safety-properties}, the AITS defined by the program graphs in \Cref{fig:PG-Peterson} is race-free and, by \Cref{prop:basic-properties-of-partially-valid-APR-problems}~\bfnum{5}, there is no normal form reachable from an initial state.
\end{example}

\begin{figure}[t]
	\def\ScoreOverhang{1pt}
	\begin{prooftree}\AxiomC{(2)\dag}
		\AxiomC{(3)\ddag}
		\BinaryInfC{(8)}
		\AxiomC{(3)\ddag}
		\RightLabel{\small(\RuleName{Der})}
		\UnaryInfC{(9)$\lozenge$}
		\RightLabel{\small(\RuleName{Der})}
		\BinaryInfC{(4)}
		\AxiomC{(9)$\lozenge$}
		\RightLabel{\small(\RuleName{Der})}
		\UnaryInfC{(5)}
		\RightLabel{\small(\RuleName{Der})}
		\BinaryInfC{(2)\dag}
		\AxiomC{(2)\dag}
		\RightLabel{\small(\RuleName{Der})}
		\UnaryInfC{(10)$\blacklozenge$}
		\RightLabel{\small(\RuleName{Der})}
		\UnaryInfC{(6)}
		\AxiomC{(10)$\blacklozenge$}
		\AxiomC{(2)\dag}
		\AxiomC{(3)\ddag}
		\BinaryInfC{(11)}
		\RightLabel{\small(\RuleName{Der})}
		\BinaryInfC{(7)}
		\RightLabel{\small(\RuleName{Der})}
		\BinaryInfC{(3)\ddag}
		\RightLabel{\small(\RuleName{Der})}
		\BinaryInfC{(1)~$\ReachPred{\{ \State{\symb{noncrit}_0,\symb{noncrit}_1,\symb{false},\symb{false},m} \mid m \in \{0,1\} \}}{\emptyset}$}
	\end{prooftree}

	\setlength{\columnsep}{3pt}
	\footnotesize
	\begin{multicols}{2}
		\begin{itemize}
			\itemsep-1pt
			\leftskip=2ex
			\item[(2)] $\ReachPred{\{\State{\symb{wait}_0,\symb{noncrit}_1,\symb{true},\symb{false},1}\}}{\emptyset}$
			\item[(3)] $\ReachPred{\{\State{\symb{noncrit}_0,\symb{wait}_1,\symb{false},\symb{true},0}\}}{\emptyset}$
			\item[(4)] $\ReachPred{\{\State{\symb{crit}_0,\symb{noncrit}_1,\symb{true},\symb{false},1}\}}{\emptyset}$
			\item[(5)] $\ReachPred{\{\State{\symb{wait}_0,\symb{wait}_1,\symb{true},\symb{true},0}\}}{\emptyset}$
			\item[(6)] $\ReachPred{\{\State{\symb{wait}_0,\symb{wait}_1,\symb{true},\symb{true},1}\}}{\emptyset}$
			\item[(7)] $\ReachPred{\{\State{\symb{noncrit}_0,\symb{crit}_1,\symb{false},\symb{true},0}\}}{\emptyset}$
			\item[(8)] $\ReachPred{\{\State{\symb{noncrit}_0,\symb{noncrit}_1,\symb{false},\symb{false},1}\}}{\emptyset}$
			\item[(9)] $\ReachPred{\{\State{\symb{crit}_0,\symb{wait}_1,\symb{true},\symb{true},0}\}}{\emptyset}$
			\item[(10)] $\ReachPred{\{\State{\symb{wait}_0,\symb{crit}_1,\symb{true},\symb{true},1}\}}{\emptyset}$
			\item[(11)] $\ReachPred{\{\State{\symb{noncrit}_0,\symb{noncrit}_1,\symb{false},\symb{false},0}\}}{\emptyset}$
		\end{itemize}
	\end{multicols}
	\caption{An APR proof for $\ReachPred{\{ \State{\symb{noncrit}_0,\symb{noncrit}_1,\symb{false},\symb{false},m} \mid m \in \{0,1\} \}}{\emptyset}$, where $\dag$, $\ddag$, $\lozenge$, and $\blacklozenge$
		indicate bud-companion relationships.
	}
	\label{fig:APR-proof-rf}
\end{figure}

For $\ReachPred{P}{Q}$ being a safety APR predicate for $E \subseteq \NF[\cA]$,
the target set $Q$ is required to satisfy that
$\{ t \in \NF[\cA] \setminus E \mid \exists s \in P.\ s \to_\cA^* t \} \subseteq Q$,
while it suffices to satisfy $\{ t \in \NF[\cA] \setminus E \mid \exists s \in P.\ s \to_\cA^* t \} = Q$.
It may be difficult for a given rewrite system $\cA$ to compute the set $\{ t \in \NF[\cA] \setminus E \mid \exists s \in P.\ s \to_\cA^* t \}$, and thus we allow to use an over-approximation $Q$ such that $\{ t \in \NF[\cA] \setminus E \mid \exists s \in P.\ s \to_\cA^* t \} \subseteq Q$.
On the other hand, as in~\cite{KN23jlamp}, a simpler modification is useful:
We introduce a fresh element $\symb{any}$ to $A$ and extend $\to_\cA$ to ${\to_\cA} \cup \{ (s,\symb{any}) \mid s \in A \setminus E \}$.
This modification is not an approximation.

\begin{restatable}{theorem}{ThmSoundnessAndCompletenessOfUsingAny}
	\label{thm:soundness-and-completeness-of-using-any}
	Let $\ReachPred{P}{Q}$ be a safety APR predicate for $E \subseteq \NF[\cA]$.
	Then,
	$\PartiallyValid[\cA]{\ReachPred{P}{Q}}$ if and only if
	$\PartiallyValid[(A\cup\{\symb{any}\}, {\to_\cA}\cup\{ (s,\symb{any}) \mid s \in A \setminus E \})]{\ReachPred{P}{\{\symb{any}\}}}$.
\end{restatable}

One may think that non-reachability analysis is more suitable for safety properties.
However, APR analysis is essentially very similar to non-reachability analysis:
\begin{itemize}
	\item For non-reachability from $P$ to $E$, all states reachable from $P$ are examined, unless an error state reachable from $P$ is detected.
	\item During the construction of an APR pre-proof, all states reachable from $P$ are examined, unless a subset $P'$ with $P'\cap E \ne \emptyset$ is detected.
\end{itemize}

\section{Reduction of Liveness Properties to APR Problems}
\label{sec:reduction-of-liveness}

Informally speaking, \emph{liveness properties} state that ``something good'' will eventually happen for every execution~\cite[Section~3.4]{BK08}.
As a liveness property w.r.t.\ a system, a source set $P$, and a set $Q$ for ``something good'', we consider the property that
any (possibly non-terminating) execution of the system starting from any state in $P$ reaches a state in $Q$.
This kind of safety properties is formulated as a problem for ARSs as follows:
\begin{description}
	\item[Instance] An ARS $\cA=(A,\to_\cA)$ and sets $P,Q$ ($\subseteq A$)
	\item[Question] Does \emph{every} (possibly infinite) execution path $s_0 \to_\cA s_1 \to_\cA \cdots$ with $s_0 \in P$ include an element in $Q$ (i.e., $s_i \in Q$ for some $i \geq 0$)?
\end{description}
Unfortunately, this problem cannot be reduced to partial validity of $\ReachPred{P}{Q}$ w.r.t.\ $\cA$, while $\ReachPred{P}{Q}$ seems natural for liveness properties.
The shortcoming is that partial validity does not take into account any infinite execution path.
In this section, we introduce a stronger validity, called \emph{total validity}, which takes into account \emph{all} (possibly infinite) execution paths.
Then, for a partially valid APR predicate with an APR proof, we show a necessary and sufficient condition for the tree to ensure total validity, showing how to apply APR analysis to verification of liveness properties.

\subsection{Total Validity of APR Predicates w.r.t.\ ARSs}
\label{subsec:total-validity}

We first introduce a stronger validity which takes into account all execution paths.
Total validity of $\ReachPred{P}{Q}$ w.r.t.\ $\cA$ ensures that every execution path starting from an element in $P$ eventually reaches an element in $Q$.
\begin{definition}[total validity]
	An APR predicate $\ReachPred{P}{Q}$ over $A$ is said to be \emph{totally valid} w.r.t.\ $\cA$, written as $\TotallyValid{\ReachPred{P}{Q}}$, if
	every (possibly infinite) execution path $s_0 \to_\cA s_1 \to_\cA \cdots$ with $s_0 \in P$ includes an element in $Q$, i.e., $s_i \in Q$ for some $i \geq 0$.
	We write $\NotTotallyValid[\cA]{\ReachPred{P}{Q}}$ if $\ReachPred{P}{Q}$ is not totally valid w.r.t.\ $\cA$.
\end{definition}

\begin{example}
	\label{ex:APR-predicate-for-PG0-starvation-freedom}
	Let us continue with \Cref{example:peterson-ec}.
	Starvation freedom for the process $P_0$ specified by $\mathit{PG}_0$ in \Cref{fig:PG-Peterson}---$P_0$ can reach $\symb{crit}_0$ from $\symb{wait}_0$---is reduced to total validity of the following APR predicate:
	\[
		(12)~
		\ReachPred{\{ \State{\symb{wait}_0,p_1,\sigma} \mid p_1 \in \StateSet{1}, \sigma \in \cS, \sigma(b_0) \,{=}\, \symb{true} \}}{\{ \State{\symb{crit}_0,p_1,\sigma} \mid p_1 \in \StateSet{1}, \sigma \in \cS \}}
	\]
\end{example}

By definition, it is clear that total validity implies partial one.
In other words, partial validity is a necessary condition for total one.
Thus, to prove total validity, we should first prove partial validity and we usually attempt to construct an APR proof.
However, the existence of such an APR proof does not ensure total validity.
For this reason, we need an additional condition.
To this end, in the next section, we will show a sufficient condition for APR proofs to additionally imply total validity.

\subsection{A Criterion for Total Validity of Partial Valid APR Predicates}
\label{subsec:proof-method-for-total-validity}

Let us consider an APR proof and its proof graph for $\ReachPred{P}{Q}$.
If the proof graph has no cycle, then there is no infinite execution path that does not include any element in $Q$;
if there exists such a path, then either it is impossible to obtain a finite APR proof or the proof graph has a cycle.
Recall that APR proofs are finite derivation trees.

\begin{restatable}{theorem}{ThmTotalValidityProof}
	\label{thm:total-validity-proof}
	Let $((V,\varsl{a},\varsl{r},\varsl{c}),\xi)$ be a partially valid APR proof for an APR predicate $\ReachPred{P}{Q}$ w.r.t.\ $\cA$,
	and $(V',E)$ be the proof graph of the APR proof. Then, $(V',E)$ is acyclic
	if and only if
	$\TotallyValid{\ReachPred{P}{Q}}$.
\end{restatable}
Once we construct an APR proof, the total validity is decidable.
\begin{corollary}\label{cor:decidability-of-total-validity-proof}
	Let $((V,\varsl{a},\varsl{r},\varsl{c}),\xi)$ be an APR proof for a partially valid APR predicate $\ReachPred{P}{Q}$ w.r.t.\ $\cA$.
	Then, total validity of $\ReachPred{P}{Q}$ w.r.t.\ $\cA$ is decidable.
\end{corollary}
\begin{proof}
	It is decidable whether a finite graph is acyclic, and thus,
	by \Cref{thm:total-validity-proof}, totally validity is decidable.
\end{proof}

We call a closed APR pre-proof \emph{acyclic} if its proof graph is acyclic.
Notice that an acyclic APR pre-proof may have a bud node and its companion.
APR pre-proofs are simpler variants of cyclic proofs and the bud-companion relationship leads to the terminology ``cyclic''.
On the other hand, the existence of bud nodes and companions in an APR pre-proof does not always induce an actual circularity:
Regarding an APR pre-proof for an APR predicate, if its proof graph is acyclic, then the APR pre-proof can be expanded to another APR pre-proof for the APR predicate, which does not include any bud node and its companion (cf.~\cite[Section~5]{Bro05}).
For such an APR pre-proof, the bud-companion relationship enables us to reduce the search space and thus the size of constructing trees.

As a consequence of \Cref{thm:total-validity-proof}, for an APR predicate $\ReachPred{P}{Q}$, the existence of an acyclic APR pre-proof is a sufficient condition for total validity of $\ReachPred{P}{Q}$ w.r.t.\ $\cA$.
\begin{corollary}
	\label{cor:sufficient-condition-for-total-validity}
	For an APR predicate $\ReachPred{P}{Q}$ over $A$,
	if there exists an acyclic APR proof for $\ReachPred{P}{Q}$ w.r.t.\ $\cA$,
	then $\TotallyValid{\ReachPred{P}{Q}}$.
\end{corollary}

\begin{example}
	\label{ex:total-validity}
	Let us continue with \Cref{ex:APR-predicate-for-PG0-starvation-freedom}.
	We have an APR proof for the APR predicate~(12) shown in \Cref{fig:APR-proof}, and the APR proof is acyclic because its proof graph shown in \Cref{fig:APR-proof-graph} is acyclic.
	Therefore, by \Cref{cor:sufficient-condition-for-total-validity}, the APR predicate (12) is totally valid w.r.t.\ $\cApet$, and thus the AITS is starvation free for $P_0$.
\end{example}

\begin{figure}[t]
	\def\defaultHypSeparation{\hskip 0.45em}
	\begin{scprooftree}{.92}
		\AxiomC{}
		\RightLabel{\small(\RuleName{Axiom})}
		\UnaryInfC{(17)}
		\RightLabel{\small(\RuleName{Subs})}
		\UnaryInfC{(13)}
		\AxiomC{}
		\RightLabel{\small(\RuleName{Axiom})}
		\UnaryInfC{(17)}
		\RightLabel{\small(\RuleName{Subs})}
		\UnaryInfC{(18)}
		\RightLabel{\small(\RuleName{Der})}
		\UnaryInfC{(14)\dag}
		\AxiomC{}
		\RightLabel{\small(\RuleName{Axiom})}
		\UnaryInfC{(17)}
		\RightLabel{\small(\RuleName{Subs})}
		\UnaryInfC{(21)}
		\AxiomC{(14)\dag}
		\RightLabel{\small(\RuleName{Der})}
		\BinaryInfC{(19)}
		\RightLabel{\small(\RuleName{Der})}
		\UnaryInfC{(15)} \AxiomC{}
		\RightLabel{\small(\RuleName{Axiom})}
		\UnaryInfC{(17)}
		\RightLabel{\small(\RuleName{Subs})}
		\UnaryInfC{(20)}
		\AxiomC{(14)\dag}
		\RightLabel{\small(\RuleName{Der})}
		\BinaryInfC{(16)}
		\RightLabel{\small(\RuleName{Der})}
		\QuaternaryInfC{(12)~$\ReachPred{\{ \State{\symb{wait}_0,p_1,\sigma} \mid p_1 \,{\in}\, \StateSet{1}, \sigma \,{\in}\, \cS, \sigma(b_0) \,{=}\, \symb{true} \}}{\{ \State{\symb{crit}_0,p_1,\sigma} \mid p_1 \,{\in}\, \StateSet{1}, \sigma \,{\in}\, \cS \}}$}
	\end{scprooftree}

	\smallskip
	\setlength{\columnsep}{3pt}
	\footnotesize
	\begin{itemize}
		\leftskip=2ex
		\itemsep-1pt
		\item[(13)]
		      $\ReachPred{\{\State{\symb{crit}_0,p_1,\symb{true},v_1,m} \mid v_1 \,{\in}\, \mathbb{B},  m \,{\in}\, \{0,1\}, m\,{=}\,0 \lor v_1\,{=}\,\symb{false} \} }{ \{ \State{\symb{crit}_0,p_1,\sigma} \mid p_1 \,{\in}\, \StateSet{1}, \sigma \,{\in}\, \cS \}}$
		\item[(14)]
		      $\ReachPred{\{\State{\symb{wait}_0,\symb{wait}_1,\symb{true},\symb{true},0}\}}{\{ \State{\symb{crit}_0,p_1,\sigma} \mid p_1 \in \StateSet{1}, \, \sigma \in \cS \}}$
		\item[(15)]
		      $\ReachPred{\{\State{\symb{wait}_0,\symb{crit}_1,\symb{true},v_1,1} \mid v_1 \in \mathbb{B} \}}{\{ \State{\symb{crit}_0,p_1,\sigma} \mid p_1 \in \StateSet{1}, \, \sigma \in \cS \}}$
		\item[(16)]
		      $\ReachPred{\{\State{\symb{wait}_0,\symb{noncrit}_1,\symb{true},\symb{false},m} \mid m \in \{0,1\}\}}{\{ \State{\symb{crit}_0,p_1,\sigma} \mid p_1 \in \StateSet{1}, \, \sigma \in \cS \}}$
		\item[(17)]
		      $\ReachPred{\emptyset}{\{ \State{\symb{crit}_0,p_1,\sigma} \mid p_1 \in \StateSet{1}, \, \sigma \in \cS \}}$
		\item[(18)]
		      $\ReachPred{\{\State{\symb{crit}_0,\symb{wait}_1,\symb{true},\symb{true},0}\}}{\{ \State{\symb{crit}_0,p_1,\sigma} \mid p_1 \in \StateSet{1}, \, \sigma \in \cS \}}$
		\item[(19)]
		      $\ReachPred{\{\State{\symb{wait}_0,\symb{noncrit}_1,\symb{true},\symb{false},1}\}}{\{ \State{\symb{crit}_0,p_1,\sigma} \mid p_1 \in \StateSet{1}, \, \sigma \in \cS \}}$
		\item[(20)]
		      $\ReachPred{\{\State{\symb{crit}_0,\symb{noncrit}_1,\symb{true},\symb{false},m} \mid m \in \{0,1\}\}}{\{ \State{\symb{crit}_0,p_1,\sigma} \mid p_1 \in \StateSet{1}, \, \sigma \in \cS \}}$
		\item[(21)]
		      $\ReachPred{\{\State{\symb{crit}_0,\symb{noncrit}_1,\symb{true},\symb{false},1}\}}{\{ \State{\symb{crit}_0,p_1,\sigma} \mid p_1 \in \StateSet{1}, \, \sigma \in \cS \}}$
		\item[(22)]
		      $\ReachPred{\{\State{\symb{wait}_0,\symb{wait}_1,\symb{true},\symb{true},0}\}}{\{ \State{\symb{crit}_0,p_1,\sigma} \mid p_1 \in \StateSet{1}, \, \sigma \in \cS \}}$
	\end{itemize}
	\caption{An APR proof for (12), where $\dag$ indicates bud-companion relationships.}
	\label{fig:APR-proof}
\end{figure}

\begin{figure}[t]
	\centering
	$
		\xymatrix@R=15pt@C=15pt{
		&              & *+[o][F]{v_{17}} \\
		& *+[o][F]{v_{17}} & *+[o][F]{v_{21}}\ar[u] && *+[o][F]{v_{17}} \\
		*+[o][F]{v_{17}} & *+[o][F]{v_{18}}\ar[u] && *+[o][F]{v_{19}}\ar[ul]\ar@(ur,ur)[dll] & *+[o][F]{v_{20}}\ar[u] \\
		*+[o][F]{v_{13}}\ar[u] & *+[o][F]{v_{14}}\ar[u] && *+[o][F]{v_{15}}\ar[u] & *+[o][F]{v_{16}}\ar[u]\ar@(ur,r)[lll] \\
		&& *+[o][F]{v_{12}}\ar[llu]\ar[lu]\ar[ru]\ar[rru] \\
		}
	$
	\caption{The proof graph of the APR proof for (12) in \Cref{fig:APR-proof}, where node $\varsl{a}(v_i)= (i)$.}
	\label{fig:APR-proof-graph}
\end{figure}

\section{Related Work}
\label{sec:related-work}

The work in this paper is the adaptation of the APR frameworks in~\cite{CL18,KN23jlamp} to ARSs, providing an abstract foundation for APR analysis.
The cyclic-proof system for partial validity of APR predicates in this paper is a simplified adaptation of the well-known cyclic-proof system $\mbox{CLKID}^\omega$~\cite[Chapter~5]{Bro06phd}.
Compared with sequents, APR predicates are very simple, and our APR proofs always satisfy the global trace condition.

Runtime-error verification by means of APR analysis w.r.t.\ constrained rewrite systems such as LCTRSs has been investigated.
$\mathbb{K}$ framework~\cite{RS10} is a more general setting of constrained rewriting than LCTRSs, and the proof system for partial validity of APR predicates consisting of constrained terms has been implemented in $\mathbb{K}$~\cite{SCMMSR14,SCMMSR19}.
The race-freedom of Peterson's mutual exclusion algorithm has been proved in~\cite{SCMMSR19} by means of the APR approach, while starvation freedom is not considered.
A comparison of \emph{all-path reachability logic} with CTL* can be seen in~\cite{SCMMSR19}.

Our previous APR-based approach in~\cite{KN23jlamp} to safety verification is to add rules to a given LCTRS so as to make any finite prefix of all (possible infinite) execution paths finite execution paths of the modified LCTRS.
On the other hand, as described in \Cref{sec:reduction-of-safety-properties-to-APR}, it suffices to include in the target set all irreducible states that are not error states, and we formulated the necessary condition as safety APR predicates for sets of error states.
However, for the inclusion of all non-error irreducible states in the target set, it is enough to add the reduction from all states to a freshly introduced dummy state that is irreducible w.r.t.\ the modified ARS.
Viewed in this light, in practice, these two approaches are similar, and the approach in this paper is an abstract framework of the previous one for LCTRSs.

Our previous work~\cite{KN23padl,KN24rp} for the verification of starvation freedom by means of APR analysis with LCTRSs reduces starvation freedom to safety APR predicates, and does not use the approach to liveness properties in this paper.
On the other hand, the previous work~\cite{KN24rp} deals with \emph{process-fairness} for starvation freedom.

\section{Conclusion}
\label{sec:conclusion}

In this paper, we first reformulated inference rules for partial validity of APR predicates w.r.t.\ ARSs, and adapted the cyclic-proof system for partial validity w.r.t.\ LCTRSs to ARSs.
The reformulated framework includes a disproof rule.
Then, we showed how to apply APR analysis to safety verification.
Finally, we introduced total validity of APR predicates w.r.t.\ ARSs and showed
that if there is an acyclic APR proof for an APR predicate w.r.t.\ an ARS, then the APR predicate is totally valid w.r.t.\ the ARS.
Total validity of APR predicates can be used for liveness verification.

Our APR pre-proofs do not consider the generalization for bud nodes.
To be more precise, in the current definition, a bud node and its companion have the same APR predicate.
As described in \Cref{fnt:relaxation-of-bud-companion-relationships}, the relaxation of the bud-companion relationship---an APR predicate can be more general than that of its companion---is one of our future directions for the practical use of APR analysis for runtime-error verification.
The introduction of \emph{process-fairness} to the abstract APR framework formulated in this paper is also a future direction to make APR-based verification more practical.
Another future direction of this research is to deal with the full version of the rule \RuleName{Circ}, i.e., ``cut'' in APR pre-proofs.

\bibliography{mybiblio}

\appendix

\section{Omitted Proofs}
\label{sec:proofs}

\PropUniquenessOfApplicableRulesInDVPnew*
\begin{proof}
	We make a case analysis depending on whether $P = \emptyset$.
	\begin{itemize}
		\item Case where $P=\emptyset$.
		      Rule \RuleName{Axiom} is applicable, but the others are not.
		\item Case where $P\ne \emptyset$.
		      We further make a case analysis depending on whether $P \cap Q = \emptyset$.
		      \begin{itemize}
			      \item Case where $P \cap Q = \emptyset$.
			            We make a case analysis depending on whether $P \cap \NF[\cA] = \emptyset$.
			            \begin{itemize}
				            \item Case where $P \cap \NF[\cA] = \emptyset$.
				                  Rule \RuleName{Der} is applicable, but the others are not.
				            \item Case where $P \cap \NF[\cA] \ne \emptyset$.
				                  There is no rule applicable to $\ReachPred{P}{Q}$.
			            \end{itemize}
			      \item Case where $P \cap Q \ne \emptyset$.
			            Rule \RuleName{Sub} is applicable, but the others are not.
		      \end{itemize}
	\end{itemize}
	Therefore, the claim holds.
\end{proof}

\ThmEquivalenceOfDVPandDVPnew*
\begin{proof}
	By definition, we have that $\nu\widehat{\DVP[\cA]} \subseteq \nu\widehat{\DVPnew[\cA]}$, because
	the applications of \RuleName{Subsumption} and \RuleName{Step} in $\DVP$ can be simulated by \RuleName{Subs} and either \RuleName{Axiom} or \RuleName{Der} in $\DVPnew$.
	We show that $\nu\widehat{\DVP[\cA]} \supseteq \nu\widehat{\DVPnew[\cA]}$.
	The application of \RuleName{Axiom} in $\DVPnew$ can be simulated by \RuleName{Subsumption} in $\DVP$.
	The application of \RuleName{Subs} in $\DVPnew$ is followed by either \RuleName{Axiom} or \RuleName{Der} in $\DVPnew$.
	Let us consider the following case:
	\begin{prooftree}
		\AxiomC{$\ReachPred{P_1}{Q}$}
		\AxiomC{\ldots}
		\AxiomC{$\ReachPred{P_n}{Q}$}
		\RightLabel{\small(\RuleName{Subs})}
		\TrinaryInfC{$\ReachPred{P}{Q}$}
	\end{prooftree}
	where
	$P \cap Q \ne \emptyset$, $P \setminus Q = P_1 \cup \cdots \cup P_n$ for some $n > 0$,
	and
	if $n > 1$ then $P_i \ne \emptyset$ for all $1 \leq i \leq n$.
	We make a case analysis depending on whether $P \setminus Q = \emptyset$ or not.
	\begin{itemize}
		\item Case where $P \setminus Q = \emptyset$.
		      By definition, we have that $P \subseteq Q$.
		      Therefore, the successive application of \RuleName{Subs} and
		      \RuleName{Axiom} can be simulated by \RuleName{Subsumption} in $\DVP$.

		\item Case where $P \setminus Q \ne \emptyset$.
		      By definition, we have that $P_i \ne \emptyset$ for all $1 \leq i \leq n$.
		      Since $P_i \subseteq P \setminus Q$, we have that $P_i \cap Q = \emptyset$.
		      Since $(\ReachPred{P}{Q}) \in \nu\widehat{\DVPnew[\cA]}$,
		      at least one rule is applicable to $\ReachPred{P_i}{Q}$ for each $i\in \{1,\ldots,n\}$.
		      By \Cref{prop:uniqueness-of-applicable-rules-in-DVPnew},
		      the applicable rule is \RuleName{Der} only.
		      Thus, for each $i \in \{1,\ldots,n\}$, we have the following subtree:
		      \begin{prooftree}
			      \AxiomC{$\ReachPred{P_{i,1}}{Q}$}
			      \AxiomC{\dots}
			      \AxiomC{$\ReachPred{P_{i,m_i}}{Q}$}
			      \RightLabel{\small(\RuleName{Der})}
			      \TrinaryInfC{$\ReachPred{P_i}{Q}$}
		      \end{prooftree}
		      where
		      $\Derivative[\cA](P_i) = P_{i,1} \cup \cdots \cup P_{i,m_i}$ for some $m_i > 0$,
		      and
		      $P_{i,j} \ne \emptyset$ for all $1 \leq j \leq m_i$.
		      By definition, we have that
		      $P \setminus Q = \bigcup_{i=1}^n P_i$
		      and
		      $P_i \setminus Q = P_i$ for all $1 \leq i \leq n$.
		      Hence, we have that
		      $\Derivative[\cA](P_i \setminus Q) = \Derivative[\cA](P_i) = \bigcup_{j=1}^{m_i} P_{i,j}$,
		      and thus
		      $\Derivative[\cA](P \setminus Q) = \Derivative[\cA](\bigcup_{i=1}^n P_i) = \bigcup_{i=1}^n \bigcup_{j=1}^{m_i} P_{i,j}$.
		      Thus, we have the following tree of $\DVP$:
		      \begin{prooftree}
			      \AxiomC{$\ReachPred{( \bigcup_{i=1}^n \bigcup_{j=1}^{m_i} P_{i,j} )}{Q}$}
			      \RightLabel{\small(\RuleName{Step})}
			      \UnaryInfC{$\ReachPred{P}{Q}$}
		      \end{prooftree}
		      Since $(\ReachPred{P_i}{Q}) \in \nu\widehat{\DVPnew[\cA]}$,
		      we have that $(\ReachPred{P_{i,j}}{Q}) \in \nu\widehat{\DVPnew[\cA]}$
		      for all $1 \leq i \leq n$ and $1 \leq j \leq m_i$.
		      By induction, we have that for all $1 \leq i \leq n$,
		      $(\ReachPred{P_{i,1}}{Q}),\ldots,(\ReachPred{P_{i,m_i}}{Q}) \in \nu\widehat{\DVP[\cA]}$,
		      and thus
		      $\PartiallyValid{\ReachPred{P_{i,j}}{Q}}$ for all $1 \leq i \leq n$ and $1 \leq j \leq m_i$.
		      It follows from \Cref{prop:basic-properties-of-partially-valid-APR-problems} that
		      $\PartiallyValid{\ReachPred{( \bigcup_{i=1}^n \bigcup_{j=1}^{m_i} P_{i,j} )}{Q}}$
		      and thus
		      $( \ReachPred{( \bigcup_{i=1}^n \bigcup_{j=1}^{m_i} P_{i,j} )}{Q} ) \in \nu\widehat{\DVP[\cA]}$.
		      Therefore, we have that $(\ReachPred{P}{Q}) \in \nu\widehat{\DVP[\cA]}$.
	\end{itemize}
	The application of \RuleName{Der} in $\DVPnew$ to $(\ReachPred{P}{Q})$ is a special case of the successive application of \RuleName{Subs} and \RuleName{Der} in $\DVPnew$, and thus $(\ReachPred{P}{Q}) \in \nu\widehat{\DVP}$.
\end{proof}

\PropUniquenessOfApplicableRulesInCDVPnew*
\begin{proof}
	By \Cref{prop:uniqueness-of-applicable-rules-in-DVPnew}, it suffices to show the case for \RuleName{Dis}.
	When a rule in $\DVPnew$ is applicable to $\ReachPred{P}{Q}$, by definition, it is clear that \RuleName{Dis} is not applicable to $\ReachPred{P}{Q}$.
	In the proof of \Cref{prop:uniqueness-of-applicable-rules-in-DVPnew}, in the case where no rule is applicable to $\ReachPred{P}{Q}$, we have that $P \ne \emptyset$, $P \cap Q = \emptyset$, and $P \cap \NF[\cA] \ne \emptyset$.
	In this case, by definition, \RuleName{Dis} is applicable.
\end{proof}

\ThmSoundnessOfAPRProofs*
\begin{proof}
	We first prove the first claim, which can be proved analogously to the proof in~\cite[Theorem~3]{CL18}.
	We proceed by contradiction.
	Assume that there exists an APR proof of $\ReachPred{P}{Q}$ and
	$\PartiallyValid{\ReachPred{P}{Q}}$.
	Let $((V,\varsl{a},\varsl{r},\varsl{c}),\xi)$ be an APR proof for $\ReachPred{P}{Q}$ with root node $v_0$, and $(V',E')$ be the proof graph of $((V,\varsl{a},\varsl{r},\varsl{c}),\xi)$ such that $v_0 \in V'$,
	where $\varsl{a}(v_0) = (\ReachPred{P}{Q})$.
	Since
	$\NotPartiallyValid{\ReachPred{P}{Q}}$,
	there exists an execution path $s_0 \to_\cA s_1 \to_\cA \cdots \to_\cA v_k$ such that
	$s_0 \in P$, $\{s_0,s_1,\ldots,s_k\} \cap Q = \emptyset$, and $s_k \in \NF[\cA]$.
	Since $\varsl{a}(v_0) = (\ReachPred{P}{Q})$, we have that $s_0 \in \varsl{a}_L(v_0)$.
	It follows from \Cref{prop:basic-properties-of-proof-graphs} that
	there exist nodes $v_1,\ldots,v_k \in V'$ such that
	$v_0 \mathrel{(\toDer \cup \toSubs)}
		v_1 \mathrel{(\toDer \cup \toSubs)}
		\cdots
		\mathrel{(\toDer \cup \toSubs)}
		v_k$.
	Since $s_k \in \NF[\cA]$, by the construction of proof graphs, $v_k$ is a node in $V$ such that $\varsl{r}(v_k) = \mbox{\rm\RuleName{Dis}}$, and thus $((V,\varsl{a},\varsl{r},\varsl{c}),\xi)$ is an APR disproof.
	This contradicts the assumption that $((V,\varsl{a},\varsl{r},\varsl{c}),\xi)$ is an APR proof.

	Next, we prove the second claim.
	Let $((V,\varsl{a},\varsl{r},\varsl{c}),\xi)$ be an APR disproof for $\ReachPred{P}{Q}$ with root node $v_0$,
	where $\varsl{a}(v_0) = (\ReachPred{P}{Q})$.
	Then, there exists node $v,v' \in V$ such that $\varsl{r}(v)=\mbox{\rm\RuleName{Dis}}$, $\varsl{c}(v)=v'$, and $\varsl{a}(v') = \bot$.
	Let $\varsl{a}(v) = (\ReachPred{P'}{Q})$.
	Then, by the condition of $v,v'$, we have that
	$P' \ne \emptyset$, $P' \cap Q = \emptyset$, and $P' \cap \NF[\cA] \ne \emptyset$.
	Thus, there exists an element $s \in P'$ such that $s \in \NF[\cA]$ and $s \notin Q$.
	By the construction of APR pre-proofs, there exist nodes $v_1,\ldots,v_k$ such that
	$v_0 \mathrel{(\toDer \cup \toSubs)}
		v_1 \mathrel{(\toDer \cup \toSubs)}
		\cdots \mathrel{(\toDer \cup \toSubs)}
		v_k \mathrel{(\toDer \cup \toSubs)}
		v \mathrel{(\toDer \cup \toSubs)}
		v'$.
	It is clear that \Cref{prop:basic-properties-of-proof-graphs} holds for paths of derivation trees, and thus, by \Cref{prop:basic-properties-of-proof-graphs}, there exist elements $s_0,s_1,\ldots,s_k$ such that
	$s_0 \mathrel{\to_{\cA}^=}
		s_1 \mathrel{\to_{\cA}^=}
		\cdots \mathrel{\to_{\cA}^=}
		s_k
		\mathrel{\to_{\cA}^=}
		s$
	and
	$s_i \in \varsl{a}_L(v_i)$ and $\varsl{a}_L(v_i) \cap Q = \emptyset$ for all $0 \leq i \leq k$.
	Thus, we have that $s_i \notin Q$ for all $0 \leq i \leq k$, and thus
	$s_0 \mathrel{\to_{\cA}^=}
		s_1 \mathrel{\to_{\cA}^=}
		\cdots \mathrel{\to_{\cA}^=}
		s_k
		\mathrel{\to_{\cA}^=}
		s$
	is a finite execution path that does not include any element in $Q$.
	Therefore,
	we have that $\NotPartiallyValid{\ReachPred{P}{Q}}$.
\end{proof}

\ThmFinitenessAndDecidabilityOfAPRPredicate*
\begin{proof}
	We first show a naive construction of APR pre-proofs of $\ReachPred{P}{Q}$, where we do not split any source sets.
	It follows from \Cref{prop:uniqueness-of-applicable-rules-in-CDVPnew} that for each APR predicate, there is exactly one applicable rule in $\CDVPnew$.
	Let $v_0$ be the root of a derivation tree we construct now, where $\varsl{a}(v_0) = (\ReachPred{P}{Q})$.
	In the following, we use $V$ as a set of nodes that have already been applied a rule in $\CDVPnew$ or considered a bud node;
	we use $\Vder$ as a set of nodes that can be companions of other nodes;
	we use $U$ as a queue of nodes that have not been applied any rule in $\CDVPnew$ yet.
	Then, we construct an APR pre-proof by breadth-first search (cf.~\cite{KN24jip}):
	\begin{enumerate}
		\item $V := \{v_0\}$, $\Vder := \emptyset$, and insert $v_0$ to $U$.
		\item Repeat the following until $U$ has no node, and then return $((V,\varsl{a},\varsl{r},\varsl{c}),\xi)$:
		      \begin{enumerate}
			      \item Take a node $v$ from $U$, and $V := V \cup \{ v \}$.
			      \item If there exists a node $v' \in \Vder$ with $\varsl{a}(v') = \varsl{a}(v)$,
			            then $\xi(v) := v'$ and skip \Bfnum{c}.
			      \item If \RuleName{Axiom} is applicable to $\varsl{a}(v)$,
			            then
			            $\varsl{r}(v) := \mbox{\RuleName{Axiom}}$, $\varsl{c}(v) := \epsilon$;
			            if \RuleName{Subs} is applicable to $\varsl{a}(v)$,
			            then
			            create a node $v'$, insert $v'$ to $U$, $\varsl{r}(v) := \mbox{\RuleName{Subs}}$,
			            $\varsl{a}(v') := (\ReachPred{(\varsl{a}_L(v)\setminus\varsl{a}_R(v))}{\varsl{a}_R(v)})$,
			            and
			            $\varsl{c}(v) := v'$;
			            if \RuleName{Der} is applicable to $\varsl{a}(v)$,
			            then
			            create a node $v'$, insert $v'$ to $U$, $\varsl{r}(v) := \mbox{\RuleName{Der}}$,
			            $\varsl{a}(v') := (\ReachPred{\Derivative[\cA](\varsl{a}_L(v))}{\varsl{a}_R(v)})$,
			            $\varsl{c}(v) := v'$,
			            and
			            $\Vder := \Vder \cup \{ v \}$;
			            if \RuleName{Dis} is applicable to $\varsl{a}(v)$,
			            then
			            create a node $v'$, $V := V \cup \{ v' \}$, $\varsl{r}(v) := \mbox{\RuleName{Dis}}$,
			            $\varsl{a}(v') := \bot$,
			            and
			            $\varsl{c}(v) := v'$.
		      \end{enumerate}
	\end{enumerate}
	It is clear that the above procedure is deterministic and the output is a closed APR pre-proof.

	To prove partial validity of $\ReachPred{P}{Q}$, it suffices to consider APR predicates $\ReachPred{P'}{Q}$ such that $P' \subseteq \Derivative[\cA]^*(P)$.
	Since $\Derivative[\cA]^*(P)$ is finite, there are only finitely many APR predicates to be considered.
	If the above procedure does not halt, then it is clear that there are infinitely many APR predicates to be considered for partial validity of $\ReachPred{P}{Q}$, and this contradicts the finiteness of APR predicates to be considered.
	If there exists a node $v$ with $\varsl{a}(v) = \bot$, then the output is an APR disproof, and otherwise, it is an APR proof.
	Therefore, the first claim holds.
	If the side conditions of rules in $\CDVPnew$ are decidable, then the procedure above halts for any APR predicate if $\Derivative[\cA]^*(P)$ is finite.
	Therefore,
	the second claim holds.
\end{proof}

\ThmPartiallyValidSafetyProperties*
\begin{proof}
	We
	first
	prove the \textit{only-if} part by contradiction.
	Assume that $\PartiallyValid{\ReachPred{P}{Q}}$ and there is a finite execution path $s_0 \to_\cA s_1 \to_\cA \cdots \to_\cA s_k$
	such that $s_0 \in P$ and $s_k \in E \subseteq \NF[\cA]$.
	Since $\ReachPred{P}{Q}$ is a safety APR predicate for $E$, we have that $Q \cap E = \emptyset$, $\{ t \in \NF[\cA] \setminus E \mid \exists s \in P.\ s \to_\cA^* t \} \subseteq Q$, and $Q \subseteq \NF[\cA]$, and hence
	$\{s_0,s_1,\ldots,s_k\} \cap Q = \emptyset$.
	This contradicts the assumption that $\PartiallyValid{\ReachPred{P}{Q}}$.

	Next, we prove the \textit{if} part by contradiction.
	Assume that
	$\NotPartiallyValid{\ReachPred{P}{Q}}$
	and there is no execution path that starts with an element in $P$ and includes an element in $E$.
	Then,
	there is a finite execution path $s_0 \to_\cA s_1 \to_\cA \cdots \to_\cA s_k \in \NF[\cA]$ such that
	$s_0 \in P$ and
	$\{s_0,s_1,\ldots,s_k\} \cap Q = \emptyset$.
	Since $\ReachPred{P}{Q}$ is a safety predicate for $E$ w.r.t.\ $\cA$, we have that $\{ t \in \NF[\cA] \setminus E \mid \exists s \in P.\ s \to_\cA^* t \} \subseteq Q$, and hence $s_k \in E$.
	This contradicts the assumption.
\end{proof}

\ThmSoundnessAndCompletenessOfUsingAny*
\begin{proof}
	Let $\cA'=(A\cup\{\symb{any}\}, \to_{\cA'})$, where ${\to_{\cA'}} = ({\to_\cA}\cup\{ (s,\symb{any}) \mid s \in A \setminus E \})$.
	Then, we have that $\NF[\cA']= E \cup \{\symb{any}\}$.

	We first
	prove the \textit{only-if} part by contradiction.
	Assume that $\PartiallyValid[\cA]{\ReachPred{P}{Q}}$ and
	$\NotPartiallyValid[\cA']{\ReachPred{P}{\{\symb{any}\}}}$.
	Then, there exists a finite execution path
	$s_0 \to_{\cA'} s_1 \to_{\cA'} \cdots \to_{\cA'} s_k \in E$ with
	$\symb{any} \notin \{s_0,s_1,\ldots,s_k\}$, and thus
	$s_0 \to_{\cA} s_1 \to_{\cA} \cdots \to_{\cA} s_k \in E$.
	It follows from \Cref{thm:partially-valid-safety-properties} that $\NotPartiallyValid[\cA]{\ReachPred{P}{Q}}$.
	This contradicts the assumption that $\PartiallyValid[\cA]{\ReachPred{P}{Q}}$.

	Next, we prove the \textit{if} part by contradiction.
	Assume that $\NotPartiallyValid[\cA]{\ReachPred{P}{Q}}$ and
	$\PartiallyValid[\cA']{\ReachPred{P}{\{\symb{any}\}}}$.
	Since $\NotPartiallyValid[\cA]{\ReachPred{P}{Q}}$, by \Cref{thm:partially-valid-safety-properties}, there exists a finite execution path $s_0 \to_\cA s_1 \to_{\cA} \cdots \to_{\cA} s_k \in E$.
	By the construction of $\cA'$, we have that $s_0 \to_{\cA'} s_1 \to_{\cA'} \cdots \to_{\cA'} s_k \in E$, and thus $\NotPartiallyValid[\cA']{\ReachPred{P}{\{\symb{any}\}}}$
	This contradicts the assumption that $\PartiallyValid[\cA']{\ReachPred{P}{\{\symb{any}\}}}$.
\end{proof}

\ThmTotalValidityProof*
\begin{proof}
	We first prove the \textit{only-if} part by contradiction.
	Assume that $(V',E)$ is acyclic and $\NotTotallyValid{\ReachPred{P}{Q}}$.
	Then, there exists a (possibly infinite) execution path $s_0 \to_\cA s_1 \to_\cA \cdots$ such that $s_0 \in P$ and none of $s_0,s_1,\ldots$ is in $Q$.
	We make a case analysis depending on whether the path is finite.
	\begin{itemize}
		\item Case where the execution path is finite.
		      Let the path have $n$ steps, i.e.,
		      $s_0 \to_\cA s_1 \to_\cA \cdots \to_\cA s_n \in \NF[\cA]$.
		      Since $s_n \in \NF[\cA] \setminus Q$,
		      there exists a node $v \in V$ such that
		      $s_n \in \varsl{a}_L(v) \subseteq A$, and thus, $\varsl{r}(v) = \RuleName{Dis}$.
		      This contradicts the assumption.

		\item Case where the execution path is infinite.
		      Since none of $s_0,s_1,\ldots$ is in $Q$,
		      there exists an infinite path of the proof graph $(V',E)$.
		      Since $((V,\varsl{a},\varsl{r},\varsl{c}),\xi)$ is an APR proof, $(V,\varsl{a},\varsl{r},\varsl{c})$ is a finite derivation tree, and thus $(V',E)$ is a finite graph having an infinite path.
		      Thus, $(V',E)$ has a cycle.
		      This contradicts the assumption.
	\end{itemize}

	Next, we prove the \textit{if} part by contradiction.
	Assume that
	$\TotallyValid{\ReachPred{P}{Q}}$
	and
	$(V',E)$ has a cycle.
	Then, there exists a bud node $v$ and its companion $v'$ in the cycle such that $\xi(v) = v'$.
	By the assumption, we have that $\varsl{a}(v)=\varsl{a}(v')$, and thus $\varsl{a}_L(v) = \varsl{a}_L(v')$. Since $v,v'$ are in the cycle, there exists a finite path from $v'$ to $v$:
	$v' \mathrel{\toDer} v_1 \mathrel{({\toDer} \cup {\toSubs})^*} v$.
	It follows from \Cref{prop:basic-properties-of-proof-graphs}~\bfnum{2}--\bfnum{3} that
	for any element $s \in \varsl{a}_L(v)$,
	there exists an element $s' \in \varsl{a}_L(v')$
	such that $s' \to_\cA^+ s$ and the reduction has no element in $Q$.
	It follows from $\varsl{a}_L(v) = \varsl{a}_L(v')$ that
	for any element $s \in \varsl{a}_L(v)$,
	there exists an element $s' \in \varsl{a}_L(v)$
	such that $s' \to_\cA^+ s$ and the reduction has no element in $Q$.
	Thus, there exists an infinite execution path that starts with an element $s'' \in \varsl{a}_L(v)$ and does not include any element in $Q$.
	In the same way, we have a reduction sequence $s_0 \to_\cA^* s''$ such that $s_0 \in P$, and thus there exists an infinite execution path that starts with an element $s'' \in \varsl{a}_L(v)$ and does not include any element in $Q$.
	This contradicts the assumption.
\end{proof}

\section{Formal Description for Instantiation of $\DVPnew$ to $\DCC$}

In this section, we show a formal description for instantiation of the proof system $\DVPnew$ for ARSs to the proof system $\DCC$ for LCTRSs.
We first recall LCTRSs in~\cite{CL18} and then show the description.

\subsection{Formal Description of LCTRSs}
\label{subsec:LCTRSs}

We briefly recall LCTRSs~\cite{FKN17tocl}, which are slightly different from those in~\cite{KN13frocos}.
Familiarity with basic notions and notations on term rewriting is assumed~\cite{BN98,Ohl02}.

Let $\cS = \{\iota_1,\ldots,\iota_n\}$ be a set of \emph{sorts} and $\Sigma=
	\{
	f_1 : \iota_{1,1} \times \cdots \times \iota_{1,k_1} \Rightarrow \iota_{1,0},
	\ldots,
	f_m : \iota_{m,1} \times \cdots \times \iota_{m,k_m} \Rightarrow \iota_{m,0}
	\}$ be an $\cS$-sorted \emph{signature}.
We let $\cV$ be an $\cS$-sorted set of variables.
A model $\cM$ for $\Sigma$ is a tuple $(M_{\iota_1},\ldots,M_{\iota_n},f_1^{\cM},\ldots,f_m^{\cM})$ such that
$M_{\iota_i}$ ($\ne \emptyset$) with $1 \leq i \leq n$ is the interpretation of $\iota_i$
and
$f_j^\cM$ with $1 \leq j \leq m$ is the interpretation of $f_j$, which is a function in $M_{\iota_{j,1}} \times \cdots \times M_{\iota_{j,k_j}} \Rightarrow M_{\iota_{j,0}}$.
We denote $\bigcup_{i=1}^n M_{\iota_i}$ by $M$.
The set of function symbols with type $\iota_{i_1} \times \cdots \times \iota_{i_k} \Rightarrow \iota$ in $\Sigma$ is denoted by $\Sigma_{\iota_{i_1},\ldots,\iota_{i_k},\iota}$.
Note that $\Sigma_{\varepsilon,\iota}$ denotes the set of constants of sort $\iota$.

We denote by $\cS^b$ a set of \emph{built-in sorts} that includes at least the sort $\sort{bool}$.
A \emph{built-in signature} $\Sigma^b$ is an $\cS^b$-sorted signature, and function symbols in $\Sigma^b$ are called \emph{built-in function symbols}.
A model for $\Sigma^b$ such that the interpretation of sort $\sort{bool}$ is $\{\top,\bot\}$ is called a \emph{built-in model} for $\Sigma^b$, which is denoted by $\cM^b$.
Regarding built-in signatures and their models, we do not distinguish $\Sigma^b_{\varepsilon,\iota}$ and $M^b_\iota$ for any $\iota \in \cS^b$.
We call a function symbol in $\Sigma^b \setminus M^b$ a \emph{calculation symbol}.
The set of first-order formulas with equality over the signature $\Sigma^b$ is denoted by $\CF^b_\Sigma$.
Formulas in $\CF^b_\Sigma$ are called \emph{built-in constraint formulas} (or simply \emph{built-in constraints}).
Functions and function symbols returning values of $\sort{bool}$ are predicates and predicate symbols, respectively, and terms of sort $\sort{bool}$ are atomic formulas.

Let $\cS$ be a set of sorts with $\cS \supseteq \cS^b$ and $\leq$ be a partial order over $\cS$.
An $(\cS,\leq)$-sorted \emph{signature modulo built-ins} is an order-sorted signature $\Sigma$ that includes $\Sigma^b$ as a subsignature and such that the only built-in constants in $\Sigma$ are elements in the built-in model, i.e., $\Sigma_{\varepsilon,\iota} = \Sigma^b_{\varepsilon,\iota}$ for any built-in sort $\iota \in \cS^b$.
We call $\Sigma^b$ a \emph{built-in subsignature} of $\Sigma$.
The \emph{constructor signature} of $\Sigma$, which is an $(\cS,\leq)$-sorted signature, is denoted by $\Sigma^c$:
$\Sigma^c = (\Sigma \setminus \Sigma^b) \cup \bigcup_{\iota \in \cS^b} \Sigma_{\varepsilon,\iota}$.
Note that $\Sigma^c$ is the signature without calculation symbols.

Regarding an $(\cS,\leq)$-sorted signature modulo built-ins $\Sigma$, we extend the built-in model $\cM^b$ to a model $\cM^\Sigma$ for $\Sigma$ defined as follows:
\begin{itemize}
	\item $M^\Sigma_\iota = \{ t \in T(\Sigma^c) \mid t : \iota \}$ for each non-built-in sort $\iota \in \cS \setminus \cS^b$,
	\item $f^{\Sigma} = f^{\Sigma^b}$ for each built-in function symbol $f \in \Sigma^b$,
	      and
	\item $f^{\Sigma}(t_1,\ldots,t_n) = f(t_1,\ldots,t_n)$ for each non-built-in function symbol $f \in \Sigma\setminus \Sigma^b$.
\end{itemize}
By fixing the interpretation of the non-built-in function symbols, constraint formulas are reduced to built-in constraint formulas by relying on an unification algorithm described in detail in~\cite{CAL18}.
We assume that $M_\iota \ne \emptyset$ for any sort $\iota \in \cS$.
Note that $M^\Sigma$ denotes the set of interpretations in in $\cM^\Sigma$:
$M^\Sigma = \bigcup_{\iota \in \cS} M^\Sigma_\iota$.
Note that $M^\Sigma = T(\Sigma^c)$ and $M^\Sigma$ (i.e., $T(\Sigma^c)$ is the set of terms that do not include any calculation symbol).
Note also that $(\Sigma^b)^c \subseteq \Sigma^c$.

The set $\CF_\Sigma$ of \emph{constraint formulas} is the set of first-order formulas with equality over the signature $\Sigma$.
As for terms, the set of variables freely occurring in a constraint formula $\phi$ is denoted by $\Var(\phi)$.
A \emph{$\Sigma$-valuation} for a set $X$ of variables is a sort-preserving mapping from $X$ to $M^\Sigma$.
Given a constraint formula $\phi$ in $\CF_\Sigma$ and a $\Sigma$-valuation $\alpha$ for $X$ with $X \supseteq \Var(\phi)$, we write $\cM^\Sigma, \alpha \models \phi$ if $\alpha(\phi)$ is evaluated by $\cM^\Sigma$ to $\top$.
The \emph{valuation semantics} of a constraint formula $\phi$ is the set $\ValSem{\phi}$ of $\Sigma$-valuations for $\Var(\phi)$ satisfying $\phi$:
$\ValSem{\phi} = \{ \alpha:\mbox{$\Sigma$-valuation for $\Var(\phi)$} \mid \cM^\Sigma, \alpha \models \phi \}$.
A constraint formula $\phi$ is called \emph{valid w.r.t.\ $\cM^\Sigma$}, written as $\cM^\Sigma \models \phi$, if
$\cM^\Sigma, \alpha \models \phi$ for any $\Sigma$-valuation $\alpha$ for $\Var(\phi)$.
A constraint formula $\phi$ is called \emph{satisfiable w.r.t.\ $\cM^\Sigma$} if
$\cM^\Sigma, \alpha \models \phi$ for some $\Sigma$-valuation $\alpha$ for $\Var(\phi)$.
Note that $\phi$ is \emph{unsatisfiable w.r.t.\ $\cM^\Sigma$} if and only if
$\cM^\Sigma, \alpha \models \neg\phi$ for any $\Sigma$-valuation $\alpha$ for $\Var(\phi)$.

A \emph{constrained term} of sort $\iota \in \cS$ is a pair $\CTerm{s}{\phi}$ of a term $s : \iota \in T(\Sigma,\cV)$ and a constraint formula $\phi$.
The \emph{state predicate semantics} of a constrained term $\CTerm{s}{\phi}$ is the set $\Eval{\CTerm{s}{\phi}}$ of ground instances of $s$ w.r.t.\ $\phi$:
$\Eval{\CTerm{s}{\phi}} = \{ \alpha(s) \mid \alpha \in \ValSem{\phi} \}$.
We naturally extend state predicate semantics to sets of constrained terms:
For a set $U$ of constrained terms,
$\Eval{U} = \bigcup_{\CTerm{s}{\phi} \in U} \Eval{\CTerm{s}{\phi}}$.

A \emph{logically constrained rewrite rule} over an $(\cS,\leq)$-sorted signature modulo built-ins $\Sigma$ is a triple $\ell \to r ~ \Constraint{\varphi}$, where
$\ell$ and $r$ are terms in $T(\Sigma,\cV)$ having the same sort,
$\ell$ is not a variable,
and $\phi$ is a constraint formula in $\CF_\Sigma$.
A \emph{logically constrained term rewrite system} (LCTRS, for short) over $\Sigma$ is a set $\cR$ of logically constrained rewrite rules over $\Sigma$.
The \emph{order-sorted rewrite relation} $\to_\cR$ over $M^\Sigma$ is defined as follows:
For any terms $s,t \in M^\Sigma$,
$s \to_\cR t$
if and only if
there exist a rule $\ell \to r ~\Constraint{\varphi} \in \cR$,
a position $p$ of $s$,
and
a $\Sigma$-valuation $\alpha$ for $\Var(\ell,r,\varphi)$
such that
$s|_p = \alpha(\ell)$,
$t = s[\alpha(r)]_p$,
and
$\cM^\Sigma, \alpha \models \varphi$.
The set of \emph{normal forms} of $\cR$ is denoted by $\NF[\cR]$.
Note that $\NF[\cR] \subseteq M^\Sigma$.
The LCTRS $\cR$ induces the ARS $(M^\Sigma,\to_\cR)$.
By definition, it is clear that
$\NF[\cR] = \NF[(M^\Sigma,\to_\cR)]$.

\begin{example}
	Let $\cSfact=\{\sort{cfg},\sort{bool},\sort{int}\}$ and $\cSfact^b = \{\sort{bool},\sort{int}\}$.
	Let $\SigmaFact$ be an $(\cSfact,\leq)$-sorted signature modulo built-ins with a built-in subsignature $\SigmaFact^b$ such that
	\begin{itemize}
		\item $\SigmaFact^b =
			      \{ \top,\bot: \sort{bool} \}
			      \cup
			      \{ n : \sort{int} \mid n \in \mathbb{Z} \}
			      \cup
			      \{
			      {\neg} : \sort{bool} \Rightarrow \sort{bool}, ~
			      {\land},{\lor},{\Rightarrow},{\Leftrightarrow}: \sort{bool} \times \sort{bool} \Rightarrow \sort{bool}
			      \}
			      \cup
			      \{
			      {+},{-},{\times},/,\symb{mod} : \sort{int} \times \sort{int} \Rightarrow \sort{int}, ~
			      {<},{\leq},{>},{\geq},{=},{\ne} : \sort{int} \times \sort{int} \Rightarrow \sort{bool}
			      \}
		      $,
		      and
		\item $\Sigma =
			      \{
			      \symb{fact} : \sort{int} \Rightarrow \sort{cfg}, ~
			      \symb{subjact} : \sort{int} \times \sort{int} \Rightarrow \sort{cfg}, ~
			      \symb{return} : \sort{int} \Rightarrow \sort{cfg}
			      \}
			      \cup \SigmaFact^b
		      $.
	\end{itemize}
	For $\SigmaFact^b$, we give a built-in model $\cMfact^b$ as follows:
	\begin{itemize}
		\item $M^b_{\sort[\scriptsize]{bool}} = \{\top,\bot\}$,
		\item $\neg^{\cMfact^b}(x) = \neg x$ for any $x \in M^b_{\sort[\scriptsize]{bool}}$,
		\item $\var{bop}^{\cMfact^b}(x_1,x_2) = x_1 \mathop{\var{bop}} x_2$ for any $x_1,x_2 \in M^b_{\sort[\scriptsize]{bool}}$ and $\var{bop} \in \SigmaFact^b_{\sort[\scriptsize]{bool},\sort[\scriptsize]{bool},\sort[\scriptsize]{bool}}$,
		\item $M^b_{\sort[\scriptsize]{int}} = \mathbb{Z}$,
		\item $\var{aop}^{\cMfact^b}(x_1,x_2) = x_1 \mathop{\var{aop}} x_2$ for any $x_1,x_2 \in M^b_{\sort[\scriptsize]{int}}$ and $\var{aop} \in \SigmaFact^b_{\sort[\scriptsize]{int},\sort[\scriptsize]{int},\sort[\scriptsize]{int}}$,
		      and
		\item $\var{cmp}^{\cMfact^b}(x_1,x_2) = x_1 \mathop{\var{cmp}} x_2$ for any $x_1,x_2 \in M^b_{\sort[\scriptsize]{int}}$ and $\var{cmp} \in \SigmaFact^b_{\sort[\scriptsize]{int},\sort[\scriptsize]{int},\sort[\scriptsize]{bool}}$.
	\end{itemize}
	The built-in model $\cMfact^b$ is extended to the model $\cMfact^\SigmaFact$ for $\SigmaFact$ as follows:
	\begin{itemize}
		\item $M^{\SigmaFact}_{\sort[\scriptsize]{bool}} = M^b_{\sort[\scriptsize]{bool}} = \{\top,\bot\}$,
		\item $M^{\SigmaFact}_{\sort[\scriptsize]{int}} = M^b_{\sort[\scriptsize]{int}} = \mathbb{Z}$,
		\item $\neg^{\cMfact}(x) = \neg^{\cMfact^b}(x)$ for any $x \in M^b_{\sort[\scriptsize]{bool}}$,
		\item $\var{bop}^{\cMfact}(x_1,x_2) = \var{bop}^{\cMfact^b}(x_1,x_2)$ for any $x_1,x_2 \in M^b_{\sort[\scriptsize]{bool}}$ and $\var{bop} \in \SigmaFact^b_{\sort[\scriptsize]{bool},\sort[\scriptsize]{bool},\sort[\scriptsize]{bool}}$,
		\item $\var{aop}^{\cMfact}(x_1,x_2) = \var{aop}^{\cMfact^b}(x_1,x_2)$ for any $x_1,x_2 \in M^b_{\sort[\scriptsize]{int}}$ and $\var{aop} \in \SigmaFact^b_{\sort[\scriptsize]{int},\sort[\scriptsize]{int},\sort[\scriptsize]{int}}$,
		\item $\var{cmp}^{\cMfact}(x_1,x_2) = \var{cmp}^{\cMfact^b}(x_1,x_2)$ for any $x_1,x_2 \in M^b_{\sort[\scriptsize]{int}}$ and $\var{cmp} \in \SigmaFact^b_{\sort[\scriptsize]{int},\sort[\scriptsize]{int},\sort[\scriptsize]{bool}}$,
		\item $M^{\SigmaFact}_{\sort[\scriptsize]{cfg}} =  \{ t \mid t : \sort{cfg} \in T(\{\symb{fact},\symb{subfact},\symb{return}\}\cup\{\top,\bot\}\cup\mathbb{Z})\}$,
		\item $\symb{fact}^{\cMfact}(x) = \symb{fact}(x)$ for any $x \in M^{\SigmaFact}_{\sort[\scriptsize]{int}}$,
		\item $\symb{subfact}^{\cMfact}(x_1,x_2) = \symb{subfact}(x_1,x_2)$ for any $x_1,x_2 \in M^{\SigmaFact}_{\sort[\scriptsize]{int}}$,
		      and
		\item $\symb{return}^{\cMfact}(x) = \symb{return}(x)$ for any $x \in M^{\SigmaFact}_{\sort[\scriptsize]{int}}$.
	\end{itemize}
	Note that $M^{\SigmaFact} = M^{\SigmaFact}_{\sort[\scriptsize]{boot}} \cup M^{\SigmaFact}_{\sort[\scriptsize]{int}} \cup M^{\SigmaFact}_{\sort[\scriptsize]{cfg}}$
	The term $\symb{fact}(1)$ is included in $M^{\SigmaFact}$, but $\symb{fact}(1+1)$ is not.
	The following LCTRS over $\SigmaFact$ calculates the \emph{factorial} function over $\mathbb{Z}$ iteratively:
	\[
		\cRfact =
		\left\{
		\begin{array}{r@{\>}c@{\>}ll}
			\symb{fact}(x)      & \to & \symb{subfact}(x,1)                                       \\
			\symb{subfact}(x,y) & \to & \symb{return}(y)                  & \Constraint{ x\leq0 } \\
			\symb{subfact}(x,y) & \to & \symb{subfact}(x - 1, x \times y) & \Constraint{ x > 0 }  \\
		\end{array}
		\right\}
	\]
	The term $\symb{fact}(3)$ is reduced by $\cRfact$ to $6$:
	$
		\symb{fact}(3)
		\to_{\cRfact}
		\symb{subfact}(3,1 )
		\to_{\cRfact}
		\symb{subfact}(2,3)
		\to_{\cRfact}
		\symb{subfact}(1,6)
		\to_{\cRfact}
		\symb{subfact}(0,6)
		\to_{\cRfact}
		\symb{return}(6)
	$.
\end{example}

In the remainder of this section, we use $\cR$ as an LCTRS over an $(\cS,\leq)$-sorted signature modulo built-ins $\Sigma$ without notice.

\subsection{All-Path Reachability Predicates of LCTRSs}

Since we formally introduced LCTRSs in \Cref{subsec:LCTRSs}, we revisit the APR framework for LCTRSs~\cite{CL18}.

\begin{definition}[\cite{CL18}]
	An \emph{APR predicate} over $\Sigma$ is a pair $\ReachPred{\CTerm{s}{\phi}}{\CTerm{t}{\psi}}$ of constrained terms $\CTerm{s}{\phi},\CTerm{t}{\psi}$, which may share variables.
	We assume w.l.o.g.\ that $\CTerm{s}{\phi}$ and $\CTerm{t}{\psi}$ have the same sort.
	We say that LCTRS $\cR$ \emph{demonically satisfies} $\ReachPred{\CTerm{s}{\phi}}{\CTerm{t}{\psi}}$ (or $\ReachPred{\CTerm{s}{\phi}}{\CTerm{t}{\psi}}$ is \emph{demonically valid w.r.t.\ $\cR$}),
	written as $\PartiallyValid[\cR]{\ReachPred{\CTerm{s}{\phi}}{\CTerm{t}{\psi}}}$,
	if
	$\PartiallyValid[(M^\Sigma,\to_\cR)]{\ReachPred{\Eval{\CTerm{\alpha(s)}{\alpha(\phi)}}}{\Eval{\CTerm{\alpha(t)}{\alpha(\psi)}}}}$
	for any $\Sigma$-valuation $\alpha$ for $\Var(s,\phi)\cap\Var(t,\psi)$.
	An \emph{execution path} of $\cR$ is an execution path of $(M^\Sigma,\to_\cR)$.
\end{definition}
Regarding APR predicates $\ReachPred{\CTerm{s}{\phi}}{\CTerm{t}{\psi}}$ over $\Sigma$, we assume w.l.o.g.\ that $\Var(s,\phi) \cap \Var(t,\psi) = \emptyset$.

\begin{theorem}
	Let $\ReachPred{\CTerm{s}{\phi}}{\CTerm{t}{\psi}}$ be an APR predicate over $\Sigma$,
	$\iota$ be the sort of $\CTerm{s}{\phi}$,
	and
	the shared variables of $\CTerm{s}{\phi}$ and $\CTerm{t}{\psi}$ be $x_1:\iota_1,\ldots,x_n:\iota_n$ ($n > 0$),
	i.e., $\Var(s,\phi) \cap \Var(t,\psi) = \{x_1,\ldots,x_n\}$.
	Let $\bar{\iota}$ be a fresh sort not in $\cS$, $c: \iota \times \iota_1 \times \cdots \times \iota_n \Rightarrow \bar{\iota}$ be an $n+1$-ary function symbol not in $\Sigma$,
	and
	$\Sigma'$ be an extension of $\Sigma$ by adding $\bar{\iota}$ and $c$ to $\cS$ and $\Sigma$, respectively:
	$\cS' = \cS \cup \{\bar{\iota}\}$ and $\Sigma' = \Sigma \cup \{ c:\iota \times \iota_1 \times \cdots \times \iota_n \Rightarrow \bar{\iota} \}$.
	Let $\cR'$ be an LCTRS over $\Sigma'$ such that the rules in $\cR'$ are the same as those in $\cR$,
	i.e.,
	$\cR$ and $\cR'$ are equivalent in the sense of including rules but they induce the ARSs $(M^\Sigma,\to_\cR)$ and $(M^{\Sigma'},\to_{\cR'})$.
	Let $x'_1:\iota_1,\ldots,x'_n:\iota_n$ be fresh pairwise different variables and $\eta$ be a renaming such that $\eta(x_i)=x'_i$ for $1\leq i \leq n$ and $\eta(y) = y$ for each variable $y \in \cV \setminus \{x_1,\ldots,x_n\}$.
	Then,
	$\PartiallyValid[\cR]{\ReachPred{\CTerm{s}{\phi}}{\CTerm{t}{\psi}}}$
	if and only if
	$\PartiallyValid[\cR']{\ReachPred{\CTerm{c(s,x_1,\ldots,x_n)}{\phi}}{\CTerm{c(\eta(t),x'_1,\ldots,x'_n)}{\eta(\psi)}}}$.\footnote{Note that $\Var(c(s,x_1,\ldots,x_n),\phi) \cap \Var(c(\eta(t),x'_1,\ldots,x'_n),\eta(\psi)) = \emptyset$.}
\end{theorem}
\begin{proof}
	We first show the \textit{only-if} part by contradiction.
	Assume that
	$\PartiallyValid[\cR]{\ReachPred{\CTerm{s}{\phi}}{\CTerm{t}{\psi}}}$
	and
	$\NotPartiallyValid[\cR']{\ReachPred{\CTerm{c(s,x_1,\ldots,x_n)}{\phi}}{\CTerm{c(\eta(t),x'_1,\ldots,x'_n)}{\eta(\psi)}}}$.
	Then, there exists a finite execution path $u_0 \to_{\cR'} u_1 \to_{\cR'} \cdots \to_{\cR'} u_m \in \NF[(M^{\Sigma'},\to_{\cR'})]$ such that
	\begin{itemize}
		\item $u_0 \in \Eval[\cM^{\Sigma'}]{\CTerm{c(s,x_1,\ldots,x_n)}{\phi}}$,
		      and
		\item $u_i \notin \Eval[\cM^{\Sigma'}]{\CTerm{c(\eta(t),x'_1,\ldots,x'_n)}{\eta(\psi)}}$ for any $i \in \{0,\ldots,m\}$.
	\end{itemize}
	Since $c$ is not defined by $\cR'$, the reduction steps of the execution path are not topmost, and thus
	\begin{itemize}
		\item there is a $\Sigma'$-valuation $\alpha$ such that $u_0 = \alpha(c(s,x_1,\ldots,x_n)) \in \Eval[\cM^{\Sigma'}]{\CTerm{c(s,x_1,\ldots,x_n)}{\phi}}$ (and thus $\cM^{\Sigma'},\alpha \models \phi$),
		\item $u_i$ is of the form $c(u'_i,\alpha(x_1),\ldots,\alpha(x_n))$ for any $i \in \{1,\ldots,m\}$,
		\item $u'_m \in \NF[(M^\Sigma,\to_\cR)]$.
	\end{itemize}
	Let $\beta$ be the $\Sigma$-valuation from $\{x_1,\ldots,x_n\}$ to $M^\Sigma$ such that $\beta(x_i)=\alpha(x_i)$ for each $i \in \{1,\ldots,n\}$.
	Let $\beta'$ be the $\Sigma$-valuation from $\{x'_1,\ldots,x'_n\}$ to $M^\Sigma$ such that $\beta'(x'_i)=\beta(x_i)$ for each $i \in \{1,\ldots,n\}$.
	Then, we have that $\beta = \beta'\circ\eta$.
	Since $\PartiallyValid[\cR]{\ReachPred{\CTerm{s}{\phi}}{\CTerm{t}{\psi}}}$,
	we have that $\PartiallyValid[(M^\Sigma,\to_\cR)]{\ReachPred{\Eval{\CTerm{\beta(s)}{\beta(\phi)}}}{\Eval{\CTerm{\beta(t)}{\beta(\psi)}}}}$, and thus
	$\PartiallyValid[(M^\Sigma,\to_\cR)]{\ReachPred{\Eval{\CTerm{\beta(s)}{\beta(\phi)}}}{\Eval{\CTerm{\beta'(\eta(t))}{\beta'(\eta(\psi))}}}}$.
	By definition, $\alpha$ is an extension of $\beta$, and thus
	$\alpha(s) \in \Eval{\CTerm{s}{\phi}} = \Eval{\CTerm{\beta(s)}{\beta(\phi)}}$.
	In addition,
	we have that
	$u'_i \notin \Eval{\CTerm{\eta(t)}{\eta(\psi)}} = \Eval{\CTerm{\beta'(\eta(t))}{\beta'(\eta(\psi))}}$ for any $i \in \{0,\ldots,m\}$:
	If $u'_i \in \Eval{\CTerm{\eta(t)}{\eta(\psi)}}$ for some $i \in \{0,\ldots,m\}$,
	then we have that $u_i \in \Eval[\cM^{\Sigma'}]{\CTerm{c(\eta(t),x'_1,\ldots,x'_n)}{\eta(\psi)}}$.
	Then, there exists a finite execution path $\alpha(s) \to_\cR u_1 \to_\cR \cdots \to_\cR u_m \in \NF[(M^\Sigma,\to_\cR)]$ such that
	$\alpha(s) \in \Eval{\CTerm{\beta(s)}{\beta(\phi)}}$ and $\alpha(s),u_1,\ldots,u_m \notin \Eval{\CTerm{\beta'(\eta(t))}{\beta'(\eta(\psi))}}$.
	This contradicts the fact that $\PartiallyValid[(M^\Sigma,\to_\cR)]{\ReachPred{\Eval{\CTerm{\beta(s)}{\beta(\phi)}}}{\Eval{\CTerm{\beta(t)}{\beta(\psi)}}}}$.

	Next, we show the \textit{if} part by contradiction.
	Assume that
	$\NotPartiallyValid[\cR]{\ReachPred{\CTerm{s}{\phi}}{\CTerm{t}{\psi}}}$
	and
	$\PartiallyValid[\cR']{\ReachPred{\CTerm{c(s,x_1,\ldots,x_n)}{\phi}}{\CTerm{c(\eta(t),x'_1,\ldots,x'_n)}{\eta(\psi)}}}$.
	Then, there is a $\Sigma$-valuation $\beta$ from $\{x_1,\ldots,x_n\}$ to $M^\Sigma$ such that
	$\NotPartiallyValid[(M^\Sigma,\to_\cR)]{\ReachPred{\Eval{\CTerm{\beta(s)}{\beta(\phi)}}}{\Eval{\CTerm{\beta(t)}{\beta(\psi)}}}}$.
	Thus, there exists a finite execution path
	$u_0 \to_\cR u_1 \to_\cR \cdots \to_\cR u_m \in \NF[(M^\Sigma,\to_\cR)]$ such that
	\begin{itemize}
		\item $u_0 \in \Eval{\CTerm{\beta(s)}{\beta(\phi)}}$,
		      and
		\item $u_i \notin \Eval{\CTerm{\beta(t)}{\beta(\psi)}}$.
	\end{itemize}
	Let $\alpha$ be a $\Sigma$-valuation for $\cV\setminus\{x_1,\ldots,x_n\}$ such that
	$\cM^\Sigma, \alpha \models \beta(\phi)$.
	Then, we have that $\cM^\Sigma, \alpha\circ\beta \models \phi$.
	By definition, we have that
	\[
		c(u_0,\beta(x_1),\ldots,\beta(x_n))
		\to_{\cR'}
		c(u_1,\beta(x_1),\ldots,\beta(x_n))
		\to_{\cR'}
		\cdots
		\to_{\cR'}
		c(u_m,\beta(x_1),\ldots,\beta(x_n))
	\]
	and
	\begin{itemize}
		\item $c(u_0,\beta(x_1),\ldots,\beta(x_n)) \in \Eval[\cM^{\Sigma'}]{\CTerm{c(\beta(s),\beta(x_1),\ldots,\beta(x_n))}{\beta(\phi)}}$,
		\item $\Eval[\cM^{\Sigma'}]{\CTerm{c(\beta(s),\beta(x_1),\ldots,\beta(x_n))}{\beta(\phi)}} \subseteq \Eval[\cM^{\Sigma'}]{\CTerm{c(s,x_1,\ldots,x_n)}{\phi}}$,
		      and
		\item $c(u_i,\beta(x_1),\ldots,\beta(x_n)) \notin \Eval[\cM^{\Sigma'}]{\CTerm{c(\beta(t),\beta(x_1),\ldots,\beta(x_n))}{\beta(\phi)}}$ for any $i \in \{0,\ldots,m\}$.
	\end{itemize}

	We now show that $c(u_i,\beta(x_1),\ldots,\beta(x_n)) \notin \Eval[\cM^{\Sigma'}]{\CTerm{c(\eta(t),x'_1,\ldots,x'_n)}{\eta(\phi)}}$ for any $i \in \{0,\ldots,m\}$.
	We proceed by contradiction.
	Assume that $c(u_i,\beta(x_1),\ldots,\beta(x_n)) \in \Eval[\cM^{\Sigma'}]{\CTerm{c(\eta(t),x'_1,\ldots,x'_n)}{\eta(\phi)}}$ for some $i \in \{0,\ldots,m\}$.
	Then, there exists a $\Sigma'$-valuation $\alpha'$ for $\Var(\eta(\phi))$ such that
	$\cM^{\Sigma'}, \alpha' \models \eta(\phi)$
	and
	$c(u_i,\beta(x_1),\ldots,\beta(x_n)) = \alpha'(c(\eta(t),x'_1,\ldots,x'_n))$,
	and thus
	$u_i = \alpha'(\eta(t)) = \alpha'(\beta(t))$
	and
	$\beta(x_i) = \alpha'(x'_i) = \alpha'(\eta(x_i))$ for each $i \in \{1,\ldots,n\}$.
	Hence, we have that
	$c(u_i,\beta(x_1),\ldots,\beta(x_n)) \in \Eval[\cM^{\Sigma'}]{\CTerm{c(\beta(t),\beta(x_1),\ldots,\beta(x_n))}{\beta(\phi)}}$.
	This contradicts the fact that
	$c(u_i,\beta(x_1),\ldots,\beta(x_n)) \notin \Eval[\cM^{\Sigma'}]{\CTerm{c(\beta(t),\beta(x_1),\ldots,\beta(x_n))}{\beta(\phi)}}$.

	In summary, we have the finite execution path
	that contradicts the assumption
	$\PartiallyValid[\cR']{\ReachPred{\CTerm{c(s,x_1,\ldots,x_n)}{\phi}}{\CTerm{c(\eta(t),x'_1,\ldots,x'_n)}{\eta(\psi)}}}$.
\end{proof}

For an APR predicate $\ReachPred{\CTerm{s}{\phi}}{\CTerm{t}{\psi}}$ with $\Var(s,\phi)\cap\Var(t,\psi)=\emptyset$,
by definition, we have the following property.
\begin{proposition}
	For an APR predicate $\ReachPred{\CTerm{s}{\phi}}{\CTerm{t}{\psi}}$ over $\Sigma$ with $\Var(s,\phi)\cap\Var(t,\psi)=\emptyset$,
	$\PartiallyValid[\cR]{\ReachPred{\CTerm{s}{\phi}}{\CTerm{t}{\psi}}}$
	if and only if
	$\PartiallyValid[(M^\Sigma,\to_\cR)]{\ReachPred{\Eval{\CTerm{s}{\phi}}}{\Eval{\CTerm{t}{\psi}}}}$.
\end{proposition}

The derivatives of LCTRSs is defined as follows.
\begin{definition}[\cite{CL18}]
	The set of \emph{derivatives} of a constrained term $\CTerm{s}{\phi}$, denoted by $\Delta_\cR(\CTerm{s}{\phi})$, is defined as follows:
	\[
		\Delta_\cR(\CTerm{s}{\phi}) = \{ \CTerm{s[r]_p}{\phi \land (s = s[\ell]_p) \land \varphi} \mid (\ell \to r ~\Constraint{\varphi}) \in \cR, \mbox{$s|_p$ is not a variable} \}
	\]
	where $\ell \to r ~\Constraint{\varphi}$ is renamed so that $\Var(s,\phi) \cap \Var(\ell,r,\varphi) = \emptyset$.
	A constrained term $\CTerm{s}{\phi}$ is called \emph{derivable w.r.t.\ $\cR$} (\emph{$\cR$-derivable}, for short) if $\Delta_\cR(\CTerm{s}{\phi}) \ne \emptyset$.
\end{definition}
The derivatives have the following property.
\begin{proposition}[\cite{CL18}]
	\label{prop:property-of-Delta}
	Let $\CTerm{s}{\phi}$ be a constrained term over $\Sigma$.
	Then, $\Eval{\Delta_\cR(\CTerm{s}{\phi})} = \Derivative[(M^\Sigma,\to_\cR)](\Eval{\CTerm{s}{\phi}})$.
\end{proposition}

We define \emph{$\cR$-runnability} of constrained terms by $\cR$-derivability and validity of a certain constraint formula~\cite{CL18}.
We call a constrained term $\CTerm{s}{\phi}$ \emph{runnable w.r.t.\ $\cR$} (\emph{$\cR$-runnable}, for short) if
$\CTerm{s}{\phi}$ is $\cR$-derivable and
the constraint formula $\phi \Rightarrow \bigvee_{i=1}^k \exists \vec{y_i}.\ \phi_i$ is valid w.r.t.\ $\cM^\Sigma$,
where $\Delta_\cR(\CTerm{s}{\phi}) = \{ \CTerm{s_1}{\phi_1},\ldots,\CTerm{s_k}{\phi_k} \}$  for some $k > 0$
and $\{\vec{y_i}\} = \Var(s_i,\phi_i) \setminus \Var(s,\phi)$ for each $i \in \{1,\ldots,k\}$.
The constraint formula has the following property.
\begin{proposition}
	\label{prop:R-runnability-properties}
	Let $\CTerm{s}{\phi}$ be a constrained term over $\Sigma$,
	$\Delta_\cR(\CTerm{s}{\phi}) = \{ \CTerm{s_1}{\phi_1},\ldots,\linebreak \CTerm{s_k}{\phi_k} \}$ for some $k \geq 0$,
	and $\{\vec{y_i}\} = \Var(s_i,\phi_i) \setminus \Var(s,\phi)$ for each $i \in \{1,\ldots,k\}$.
	Then,
	$\phi \Rightarrow \bigvee_{i=1}^k \exists \vec{y_i}.\ \phi_i$ is valid w.r.t.\ $\cM^\Sigma$
	if and only if
	$k\geq 1$ and
	$\Eval{\CTerm{s}{\phi}} \cap \NF[\cR] = \emptyset$.
\end{proposition}
\begin{proof}
	We first show the \textit{only-if} part.
	It suffices to show that $k \geq 1$ and every term in $\Eval{\CTerm{s}{\phi}}$ is not a normal form of $\cR$.
	Let $u \in \Eval{\CTerm{s}{\phi}}$.
	Then, by definition, there exists a $\Sigma$-valuation $\alpha$ for $\Var(s,\phi)$ such that
	$\cM^\Sigma, \alpha \models \phi$.
	By assumption, we have that
	$\cM^\Sigma, \alpha \models \phi \Rightarrow \bigvee_{i=1}^k \exists \vec{y_i}.\ \phi_i$, and thus
	$k\geq 1$ and $\cM^\Sigma, \alpha \models \exists \vec{y_i}.\ \phi_i$ for some $i \in \{1,\ldots,k\}$.
	Let $p$ be a position of $s$ with $s|_p \notin\cV$, $s_i$ be $s[r]_p$, and $\phi_i$ be $\phi \land (s = s[\ell]_p) \land \varphi$, where $\ell \to r ~ \Constraint{\varphi} \in \cR$ and $\Var(s,\phi) \cap \Var(\ell,r,\varphi) = \emptyset$.
	Then, by definition, we have that $\{\vec{y_i}\} = \Var(\ell,r,\varphi)$.
	Since $\cM^\Sigma, \alpha \models \exists \vec{y_i}.\ \phi_i$, we have that
	$\cM^\Sigma, \alpha \models \exists \vec{y_i}.\ (\phi \land (s = s[\ell]_p) \land \varphi)$, and thus
	$\cM^\Sigma, \alpha \models \exists \vec{y_i}.\ ((s = s[\ell]_p) \land \varphi)$.
	Hence we have that
	$\cM^\Sigma \models \exists \vec{y_i}.\ ((\alpha(s) = \alpha(s)[\ell]_p) \land \varphi)$.
	Then, there exists a $\Sigma$-valuation $\beta$ for $\{\vec{y_i}\}$ such that
	$\cM^\Sigma, \beta \models (\alpha(s) = \alpha(s)[\ell]_p) \land \varphi$.
	We now have that $\alpha(s) = \alpha(s)[\beta(\ell)]_p$ and $\cM^\Sigma, \beta \models \varphi$, and thus
	$u = \alpha(s) \to_\cR \alpha(s[\beta(r)]_p$.
	Therefore, $u$ is not a normal form of $\cR$.

	Next, we show the \textit{if} part by contradiction.
	Assume that
	$k\geq 1$,
	$\Eval{\CTerm{s}{\phi}} \cap \NF[\cR] = \emptyset$,
	and
	$\phi \Rightarrow \bigvee_{i=1}^k \exists \vec{y_i}.\ \phi_i$ is not valid w.r.t.\ $\cM^\Sigma$.
	Then, there exists a $\Sigma$-valuation $\alpha$ for $\Var(s,\phi)$ such that
	$\cM^\Sigma, \alpha \not\models \phi \Rightarrow \bigvee_{i=1}^k \exists \vec{y_i}.\ \phi_i$, and thus
	$\cM^\Sigma, \alpha \models \phi$ and
	$\cM^\Sigma, \alpha \not\models \bigvee_{i=1}^k \exists \vec{y_i}.\ \phi_i$.
	By definition, we have that $\alpha(s) \in \Eval{\CTerm{s}{\phi}}$.
	By assumption, we have that $\alpha(s) \notin \NF[\cR]$.
	Then, by definition, there exist a rule $\ell \to r ~\Constraint{\varphi} \in \cR$, a position $p$ of $\alpha(s)$, and a $\Sigma$-valuation $\beta$ from $\Var(\ell,r,\varphi)$ to $M^\Sigma$ such that
	$\alpha(s)|_p \notin \cV$,
	$\alpha(s)|_p = \beta(\ell)$,
	$\cM^\Sigma, \beta \models \varphi$,
	and
	$\alpha(s) \to_\cR \alpha(s)[\beta(r)]_p$.
	By definition, we have that $\alpha(s)[\beta(r)]_p \in \Derivative[(M^\Sigma,\to_\cR)](\CTerm{s}{\phi})
		= \Eval{\Delta_\cR(\CTerm{s}{\phi})}
	$ (by \Cref{prop:property-of-Delta}), and thus
	there exists some $i \in \{1\leq k\}$ such that
	$\alpha(s)[\beta(r)]_p \in \Eval{\CTerm{s_i}{\phi_i}}$.
	This implies that $\cM^\Sigma, \alpha \models \exists \vec{y_i}.\ \phi_i$.
	This contradicts the fact that $\cM^\Sigma, \alpha \not\models \bigvee_{i=1}^k \exists \vec{y_i}.\ \phi_i$.
\end{proof}

We now recall the proof system $\DCC$.
Since we assumed that $\Var(s,\phi) \cap \Var(t,\psi) = \emptyset$ for any APR predicate $\ReachPred{\CTerm{s}{\phi}}{\CTerm{t}{\psi}}$,
the inference rules of $\DCC$ in \Cref{fig:rules-of-DCC} can be simplified as shown in \Cref{fig:rules-of-simplified-DCC}.
Note that the rules in \Cref{fig:rules-of-simplified-DCC} are adapted to the LCTRS formalism in \Cref{subsec:LCTRSs} and rule \RuleName{Circ} is replaced by \RuleName{Cyc} on page~\pageref{eq:Cyc-for-LCTRS}.

\begin{figure}[t]
	\[
		\InfRule[Axiom]{}{\ReachPred{\CTerm{s}{\phi}}{\CTerm{t}{\psi}}
		}
		~\mbox{if $\phi$ is unsatisfiable w.r.t.\ \textcolor{blue}{$\cM^\Sigma$}.}
	\]
	\[
		\InfRule[Subs]{\ReachPred{\CTerm{s}{\phi \land \neg (\exists \vec{x}.\ ((s = t) \land \psi))}}{\CTerm{t}{\psi}}
		}{\ReachPred{\CTerm{s}{\phi}}{\CTerm{t}{\psi}}
		}
		~\mbox{if $(s = t) \land \phi \land \psi$ is satisfiable w.r.t.\ \textcolor{blue}{$\cM^\Sigma$},}
	\]
	\qquad
	where \textcolor{blue}{$\{\vec{x}\} = \Var(t,\psi)$}.
	\[
		\InfRule[Der]{\ReachPred{\CTerm{s_1}{\phi_1}}{\CTerm{t}{\psi}}
			\quad
			\ldots
			\quad
			\ReachPred{\CTerm{s_n}{\phi_n}}{\CTerm{t}{\psi}}
		}{\ReachPred{\CTerm{s}{\phi}}{\CTerm{t}{\psi}}}
		~\mbox{if
			$\CTerm{s}{\phi}$ is $\cR$-runnable,
		}
	\]
	\qquad
	where $\Delta_\cR(\CTerm{s}{\phi}) = \{ \CTerm{s_i}{\phi_i} \mid 1 \leq i \leq n \}$ for some $n>0$.
	\[\color{blue}
		\InfRule[Cyc]{}{\ReachPred{\CTerm{s}{\phi}\!}{\!\CTerm{t}{\psi}}
		}
		~\mbox{if
			$\CTerm{s}{\phi} \,{\subseteq}\, \CTerm{s'\!}{\phi'}$
			and
			$\CTerm{t}{\psi} \,{\supseteq}\, \CTerm{t'\!}{\psi'}$
			for some $(\ReachPred{\CTerm{s'\!}{\phi'}\!}{\!\CTerm{t'\!}{\psi'}}) \,{\in}\, G$.}
	\]
	\caption{Simplified inference rules of $\DCC$~\cite{CL18} for partial validity of APR predicates in $G$ of an LCTRS $\cR$ over $\Sigma$,
		where the differences from \Cref{fig:rules-of-DCC} are highlighted in blue.}
	\label{fig:rules-of-simplified-DCC}
\end{figure}

\subsection{Formal Description for Instantiation}

We first show some properties of constrained terms and their state predicate semantics.
\begin{proposition}
	\label{prop:properties-of-side-conditions-of-DCC}
	Let $\CTerm{s}{\phi}$ and $\CTerm{t}{\psi}$ be constrained terms over $\Sigma$ such that $\Var(s,\phi) \cap \Var(t,\psi)=\emptyset$.
	Let $\{\vec{x}\} = \Var(t,\psi)$ and $\Delta_\cR(\CTerm{s}{\phi}) = \{ \CTerm{s_1}{\phi_1},\ldots,\CTerm{s_n}{\phi_n} \}$ for some $n \geq 0$.
	Then, all of the following statements hold:
	\begin{enumerate}
		\renewcommand{\labelenumi}{(\arabic{enumi})}
		\leftskip=1ex
		\item
		      $\phi$ is unsatisfiable w.r.t.\ $\cM^\Sigma$
		      if and only if
		      $\Eval{\CTerm{s}{\phi}} = \emptyset$,

		\item $(s = t) \land \phi \land \psi$ is satisfiable w.r.t.\ $\cM^\Sigma$
		      if and only if
		      $\Eval{\CTerm{s}{\phi}} \cap \Eval{\CTerm{t}{\psi}} \ne \emptyset$,

		\item
		      $\Eval{\CTerm{s}{\phi \land \neg (\exists \vec{x}.\ ((s = t) \land \psi))}}
			      =
			      \Eval{\CTerm{s}{\phi}} \setminus \Eval{\CTerm{t}{\psi}}
		      $,

		\item $\CTerm{s}{\phi}$ is $\cR$-runnable
		      if and only if
		      $\Eval{\CTerm{s}{\phi}} \ne \emptyset$
		      and
		      $\Eval{\CTerm{s}{\phi}} \cap \NF[(M^\Sigma,\to_\cR)] = \emptyset$,
		      and

		\item
		      $\Eval{\Delta_\cR(\CTerm{s}{\phi})}
			      =
			      \Derivative[(M^\Sigma,\to_\cR)](\Eval{\CTerm{s}{\phi}})
			      =
			      \Eval{\CTerm{s_1}{\phi_1}} \cup \cdots \cup \Eval{\CTerm{s_n}{\phi_n}}
		      $.
	\end{enumerate}
\end{proposition}
\begin{proof}
	The first claim~\Bfnum{(1)} is trivial by definition.
	We prove the remaining claims.
	\begin{enumerate}
		\renewcommand{\labelenumi}{(\arabic{enumi})}
		\leftskip=1ex
		\item[\Bfnum{(2)}]
		      We first show the \textit{only-if} part.
		      Assume that $(s = t) \land \phi \land \psi$ is satisfiable w.r.t.\ $\cM^\Sigma$.
		      Then, by definition, there exists a $\Sigma$-valuation $\alpha$ for $\Var(s,\phi,t,\psi)$ such that
		      $\cM^\Sigma, \alpha \models (s = t) \land \phi \land \psi$.
		      Thus, we have that
		      $\alpha(s) = \alpha(t)$,
		      $\cM^\Sigma, \alpha \models \phi$, and
		      $\cM^\Sigma, \alpha \models \psi$.
		      By definition, we have that
		      $\alpha(s) \in \Eval{\CTerm{s}{\phi}}$
		      and
		      $\alpha(t) \in \Eval{\CTerm{t}{\psi}}$.
		      Therefore, we have that $\alpha(s) \in \Eval{\CTerm{s}{\phi}} \cap \Eval{\CTerm{t}{\psi}}$,
		      and hence $\Eval{\CTerm{s}{\phi}} \cap \Eval{\CTerm{t}{\psi}} \ne \emptyset$.

		      Next, we show the \textit{if} part.
		      Assume that $\Eval{\CTerm{s}{\phi}} \cap \Eval{\CTerm{t}{\psi}} \ne \emptyset$.
		      Then, there exists a term $u \in \Eval{\CTerm{s}{\phi}} \cap \Eval{\CTerm{t}{\psi}}$.
		      By definition, there exist $\Sigma$-valuations $\alpha_1,\alpha_2$ for $\Var(s,\phi)$ and $\Var(t,\psi)$, respectively, such that
		      $u= \alpha_1(s)$,
		      $\cM^\Sigma, \alpha_1 \models \phi$,
		      $u= \alpha_2(t)$,
		      and
		      $\cM^\Sigma, \alpha_2 \models \psi$.
		      By assumption, we have that $\Var(s,\phi)\cap\Var(t,\psi)=\emptyset$.
		      Let $\alpha$ be the $\Sigma$-valuation for $\Var(s,\phi,t,\psi)$ such that
		      $\alpha(x) = \alpha_1(x)$ for any variable $x$ in $\Var(s,\phi)$, and
		      $\alpha(x) = \alpha_2(x)$ for any variable $x$ in $\Var(t,\psi)$.
		      Then, we have that
		      $\alpha(s) = \alpha_1(s) = u = \alpha_2(t) = \alpha(t)$,
		      $\cM^\Sigma, \alpha \models \phi$, and
		      $\cM^\Sigma, \alpha \models \psi$, and thus
		      $\cM^\Sigma, \alpha \models (s = t) \land \phi \land \psi$.
		      Therefore, $(s = t) \land \phi \land \psi$ is satisfiable w.r.t.\ $\cM^\Sigma$.

		\item[\Bfnum{(3)}]
		      We first show that
		      $\Eval{\CTerm{s}{\phi \land \neg (\exists \vec{x}.\ ((s = t) \land \psi))}}
			      \subseteq
			      \Eval{\CTerm{s}{\phi}} \setminus \Eval{\CTerm{t}{\psi}}
		      $.
		      We proceed by contradiction.
		      Assume that
		      $\Eval{\CTerm{s}{\phi \land \neg (\exists \vec{x}.\ ((s = t) \land \psi))}}
			      \not\subseteq
			      \Eval{\CTerm{s}{\phi}} \setminus \Eval{\CTerm{t}{\psi}}
		      $.
		      Then, there exists a term $u$ such that
		      $u \in \Eval{\CTerm{s}{\phi \land \neg (\exists \vec{x}.\ ((s = t) \land \psi))}}$
		      but
		      $u \notin \Eval{\CTerm{s}{\phi}} \setminus \Eval{\CTerm{t}{\psi}}$.
		      By definition, there exists a $\Sigma$-valuation $\alpha$ for $\Var(s,\phi)$ such that
		      $u = \alpha(s)$ and
		      $\cM^\Sigma, \alpha \models \phi \land \neg (\exists \vec{x}.\ ((s = t) \land \psi))$.
		      We have that
		      $\cM^\Sigma, \alpha \models \phi$
		      and
		      $\cM^\Sigma, \alpha \models \neg (\exists \vec{x}.\ ((s = t) \land \psi))$, and thus,
		      $u \in \Eval{\CTerm{s}{\phi}}$.
		      By the assumption, we have that $u \in \Eval{\CTerm{t}{\psi}}$, and thus
		      there exists a $\Sigma$-valuation $\beta$ for $\Var(t,\psi)$ such that
		      $u = \beta(t)$ and
		      $\cM^\Sigma, \beta \models \psi$.
		      Since $\cM^\Sigma, \alpha \models \neg (\exists \vec{x}.\ ((s = t) \land \psi))$
		      and the domain of $\alpha$ is $\Var(s,\phi)$ with $\Var(s,\phi) \cap \Var(t,\psi) = \emptyset$,
		      we have that $\cM^\Sigma \models \neg (\exists \vec{x}.\ ((\alpha(s) = t) \land \psi))$, and thus
		      $\cM^\Sigma \not\models \exists \vec{x}.\ ((\alpha(s) = t) \land \psi)$.
		      This implies that there is no $\Sigma$-valuation $\gamma$ for $\Var(t,\psi)$ such that
		      $\cM^\Sigma, \gamma \models (\alpha(s) = t) \land \psi$.
		      Since $\alpha(s)=u=\beta(t)$,
		      we have that $\cM^\Sigma, \beta \models (\alpha(s) = t) \land \psi$.
		      This contradicts the non-existence of $\Sigma$-valuations $\gamma$ such that $\cM^\Sigma, \gamma \models (\alpha(s) = t) \land \psi$.

		      Next, we show that
		      $\Eval{\CTerm{s}{\phi \land \neg (\exists \vec{x}.\ ((s = t) \land \psi))}}
			      \supseteq
			      \Eval{\CTerm{s}{\phi}} \setminus \Eval{\CTerm{t}{\psi}}
		      $.
		      We proceed by contradiction.
		      Assume that
		      $\Eval{\CTerm{s}{\phi \land \neg (\exists \vec{x}.\ ((s = t) \land \psi))}}
			      \not\supseteq
			      \Eval{\CTerm{s}{\phi}} \setminus \Eval{\CTerm{t}{\psi}}
		      $.
		      Then, there exists a term $u$ such that
		      $u \in \Eval{\CTerm{s}{\phi}}$,
		      $u \notin \Eval{\CTerm{t}{\psi}}$,
		      and
		      $u \notin \Eval{\CTerm{s}{\phi \land \neg (\exists \vec{x}.\ ((s = t) \land \psi))}}$.
		      By definition, there exists a $\Sigma$-valuation $\alpha$ for $\Var(s,\phi)$ such that
		      $\cM^\Sigma, \alpha \models \phi$.
		      Since $\cM^\Sigma, \alpha \models \phi$ and $u \notin \Eval{\CTerm{s}{\phi \land \neg (\exists \vec{x}.\ ((s = t) \land \psi))}}$,
		      we have that $\cM^\Sigma, \alpha \not\models \neg (\exists \vec{x}.\ ((s = t) \land \psi))$, and thus
		      $\cM^\Sigma, \alpha \models \exists \vec{x}.\ ((s = t) \land \psi)$.
		      Then, there exists a $\Sigma$-valuation $\beta$ for $\Var(t,\psi)$ such that
		      $\cM^\Sigma, \beta \models (\alpha(s)=t) \land \psi$, and thus
		      $u \in  \Eval{\CTerm{t}{\psi}}$.
		      This contradicts the fact that $u \notin \Eval{\CTerm{t}{\psi}}$.

		\item[\Bfnum{(4)}]
		      Trivial by definition and \Cref{prop:R-runnability-properties}.

		\item[\Bfnum{(5)}]
		      Trivial by definition and \Cref{prop:property-of-Delta}.

	\end{enumerate}
\end{proof}

By using constrained terms as sets of target objects, the APR framework for LCTRSs is formulated and $\DCC$ is defined for constrained terms.
Since an LCTRS $\cR$ induces the ARS $(M^\Sigma,\to_\cR)$,
by taking it as an ARS $\cA$ for $\DVPnew$,
\Cref{prop:properties-of-side-conditions-of-DCC} implies that each rules in $\DCC$ is an instance of the corresponding one in $\DVPnew$.
Therefore, $\DCC$ can be considered an instance of $\DVPnew$.

\end{document}